\documentclass[a4paper,10pt]{article}

\usepackage{amssymb}
\usepackage{latexsym}
\usepackage{amsmath}
\usepackage{euscript}

\newcommand\gH{{\mathfrak{H}}}
\newcommand\gotH{{\mathfrak{H}}}
\newcommand\goth{{\mathfrak{h}}}

\newcommand\gotL{{\mathfrak{L}}}

\newcommand\gotm{{\mathfrak{m}}}

\newcommand\gotN{{\mathfrak{N}}}
\newcommand\gotn{{\mathfrak{n}}}

\newcommand\gotS{{\mathfrak{S}}}

\newcommand\gt{{\mathfrak{t}}}

\newcommand{\gd}{{\delta}}

\newcommand{\gga}{{\gamma}}
\newcommand{\gG}{{\Gamma}}

\newcommand{\gl}{{\lambda}}

\newcommand{\gs}{{\sigma}}

\newcommand{\gth}{{\theta}}
\newcommand{\gT}{{\Theta}}

\newcommand{\gP}{{\Pi}}

\newcommand\R{{\mathbb{R}}}
\newcommand\C{{\mathbb{C}}}

\newcommand\N{{\mathbb{N}}}

\newcommand\cA{{\mathcal{A}}}

\newcommand\cC{{\mathcal{C}}}
\newcommand\cD{{\mathcal{D}}}

\newcommand\cF{{\mathcal{F}}}
\newcommand\cG{{\mathcal{G}}}
\newcommand\cH{{\mathcal{H}}}
\newcommand\cK{{\mathcal{K}}}
\newcommand\cL{{\mathcal{L}}}

\newcommand\cT{{\mathcal{T}}}

\def\wt#1{{{\widetilde #1} }}
\def\wh#1{{{\,\widehat #1\,} }}

\def\bm\chi{\mbox{\boldmath$\chi$}}

\def\IM{{\rm Im\,}}
\def\real{{\rm Re\,}}
\def\imag{{\rm Im\,}}

\def\Ext{{\rm Ext\,}}

\def\ess{{\rm ess\,}}

\def\ran{{\rm ran\,}}
\def\dom{{\rm dom\,}}

\def\dim{{\rm dim\,}}
\def\diag{{\rm diag\,}}
\def\graph{{\rm gr\,}}

\let\xker=\ker \def\ker{{\xker\,}}

\def\supp{{\rm supp\,}}

\def\dim{{\rm dim}}
\def\cl{{\rm cl}}

\DeclareMathOperator\re{Re}
\DeclareMathOperator\im{Im}

\newtheorem{theorem}{Theorem}[section]
\newtheorem{proposition}[theorem]{Proposition}
\newtheorem{corollary}[theorem]{Corollary}
\newtheorem{lemma}[theorem]{Lemma}
\newtheorem{definition}[theorem]{Definition}

\newtheorem{remark}[theorem]{Remark}

\numberwithin{equation}{section}
\newcommand{\ba}{\begin{array}}
\newcommand{\ea}{\end{array}}
\newcommand{\bea}{\begin{eqnarray}}
\newcommand{\eea}{\end{eqnarray}}
\newcommand{\bead}{\begin{eqnarray*}}
\newcommand{\eead}{\end{eqnarray*}}
\newcommand{\be}{\begin{equation}}
\newcommand{\ee}{\end{equation}}
\newcommand{\bed}{\begin{displaymath}}
\newcommand{\eed}{\end{displaymath}}
\newcommand{\bl}{\begin{lemma}}
\newcommand{\el}{\end{lemma}}
\newcommand{\bp}{\begin{proposition}}
\newcommand{\ep}{\end{propostion}}
\newcommand{\bt}{\begin{theorem}}
\newcommand{\et}{\end{theorem}}
\newcommand{\Label}{\label}
\newcommand{\bc}{\begin{corollary}}
\newcommand{\ec}{\end{corollary}}
\newcommand{\la}{\Label}

\newcommand{\br}{\begin{remark}}
\newcommand{\er}{\end{remark}}
\newcommand{\bd}{\begin{definition}}
\newcommand{\ed}{\end{definition}}

\newcommand{\slim}{\,\mbox{\rm s-}\hspace{-2pt} \lim}
\newcommand{\wlim}{\,\mbox{\rm w-}\hspace{-2pt} \lim}

\newenvironment{proof}%
{\begin{sloppypar}\noindent{\bf Proof.}}%
{\hspace*{\fill}$\square$\end{sloppypar}}

\begin{document}
\author{Mark~Malamud\\
IAMM, NAS of Ukraine\\
Universitetskaya str. 74,
83114 Donetsk, Ukraine\\[1mm]
and\\[1mm]
Donetsk National University\\
Universitetskaya Str. 24,
83050 Donetsk,  Ukraine\\[1mm]
E-Mail: mmm@telenet.dn.ua\\[3mm]
\and
Hagen~Neidhardt\\
Weierstrass Institute\\
Mohrenstr. 39,
10117 Berlin, Germany\\[1mm]
E-Mail: hagen.neidhardt@wias-berlin.de
}

\title{Sturm-Liouville boundary value problems\\
with operator potentials and \\
unitary equivalence}

\maketitle

\begin{abstract}
\noindent
Consider the minimal Sturm-Liouville operator $A =
A_{\rm min}$  generated by the differential expression
\bed
\cA := -\frac{d^2}{dt^2} + T
\eed
in the Hilbert space $L^2(\R_+,\cH)$ where $T = T^*\ge 0$ in $\cH$.
We investigate the absolutely continuous parts
of different self-adjoint realizations of $\cA$. In particular, we show
that Dirichlet and Neumann realizations, $A^D$ and $A^N$, are absolutely
continuous and unitary equivalent to each other and to the
absolutely continuous part of the Krein realization. Moreover,
if $\inf\sigma_{\ess}(T) = \inf\gs(T) \ge 0$, then the
part $\wt A^{ac}E_{\wt A}(\gs(A^D))$ of any self-adjoint realization
$\wt A$ of $\cA$ is unitarily equivalent to $A^D$.
In addition, we  prove that the absolutely continuous part
$\wt A^{ac}$ of any realization $\wt A$  is unitarily
equivalent to $A^D$  provided that the resolvent difference
$(\wt A - i)^{-1}- (A^D - i)^{-1}$ is compact. The abstract
results are applied to elliptic  differential expression in the half-space.\\


\end{abstract}

\newpage

\tableofcontents

\section{Introduction} \label{intro}

Let $T$ be a non-negative self-adjoint operator in an infinite
dimensional separable Hilbert space $\cH$. We consider the
minimal Sturm-Liouville operator $A$ generated by the differential expression
\be\la{8.10}
\cA := -\frac{d^2}{dt^2} + T
\ee
in the Hilbert space $\gotH := L^2(\R_+,\cH)$ of
$\cH$-valued square summable  vector-valued functions.
Following \cite{Gor71,GG91} the minimal operator $A := A_{\rm min}$
is defined as the closure of the operator $A'$ defined by
\be \label{8.11}
A'  :=  \cA\upharpoonright \cD_0,
\quad \cD_0 :=
\left\{\sum_{1\le j \le k}\phi_j(t)h_j:
\begin{matrix}
\phi_j \in W_0^{2,2}(\R_+)\\
 h_j \in \dom(T), \;\  k \in \N
\end{matrix}
\right\},
\ee
where  $W_0^{2,2}(\R_+) := \{\phi \in W^{2,2}(\R_+): \phi(0) =
\phi'(0) = 0\}$, that is, $A_{\rm min} := \overline{A'}$. It
is easily seen that  $A$ is a closed non-negative symmetric
operator in $\cH$ with equal deficiency indices $n_{\pm}(A)=
\dim(\cH)$. The adjoint operator $A^*$
of $A = A_{\rm min}$ is the maximal operator denoted by
$A_{\rm max}$
Extensions of $A$ are usually called realizations of
{$\cA$, self-adjoint extensions are called self-adjoint
realizations.} Self-adjoint realizations of {$\cA$}
were firstly investigated by M. L. Gorbachuk \cite{Gor71} {in
the case of finite intervals $I$}. Namely, he  showed that the
traces of vector-functions $f\in \dom(A_{\rm max})$ belong to the
space $\cH_{-1/4}(T)$, cf. \eqref{5.1b}.  In particular, $\dom(A_{\rm max})$  is not
contained in the Sobolev space $W^{2,2}(I,\cH)$.
Based on this result he constructed a boundary triplet for the operator
$A_{\rm max}= A_{\rm min}^* = A^*$ in the Hilbert space $L^2(I,\cH)$.
These results are similar to those for elliptic operators in domains
with smooth boundaries, cf. \cite{Ber65,Gru08,LioMag72}, and go back to classical papers
of M.I. Vi\v{s}ik \cite{Visik52} and G. Grubb \cite{Grubb68}.

After the pioneering work \cite{Gor71} the spectral theory of
self-adjoint and dissipative realizations of {$\cA$} in
$L^2(I,\cH)$ has intensively been  investigated by several authors
{for bounded intervals}. Their results have been summarized in the
book of M.L. and V.I. Gorbachuk \cite[Section 4]{GG91} where one
finds, in particular, discreteness criterion, asymptotic formulas
for the eigenvalues, resolvent comparability results, etc. Some
results from \cite{GG91} including the construction of a boundary
triplet were extended in  \cite{GorKut78},  \cite{GorKut82},
\cite{Kut76},  \cite[Section 9]{DM91},  to the case of the
semi-axis.  In particular, in \cite{GorKut78},  \cite{GorKut82},
the $\mathfrak S_p$-resolvent comparability of two realizations of the
form $y'(0) = C_jy(0), \  j\in \{1,2\}$. For instance, the
Dirichlet and the Neumann realizations are $\mathfrak S_1$-resolvent
comparable if and only if $T^{-1}\in\mathfrak S_1$ (cf.
\cite{GorKut78}).

However neither the absolutely continuous spectrum ({in short}
$ac$-spectrum) nor the unitary equivalence of self-adjoint realizations of
$\cA$ have been investigated in previous papers.
We show, cf. Lemma \ref{VI.2.7}, that the domain $\dom(A)$
of the minimal operator $A$ coincides algebraically and
topologically with the Sobolev space $W^{2,2}_{0,T}(\R_+, \cH) :=
\{f\in W^{2,2}_T(\R_+, \cH): \  f(0) = f'(0)=0\},$
where $W^{2,2}_T(\R_+, \cH)$ consists of $\cH$-valued
functions $f(\cdot)\in W^{2,2}(\R_+, \cH)$  satisfying
\bed
\|f\|^2_{W_T^{2,2}} := \int_{{\R}_+}\bigl(\|f''(t)\|^2_{\cH} +
\|f(t)\|^2_\cH +\|T f(t)\|_{\cH}^2\bigr)dt <\infty.
\eed
This statement is similar to the classical regularity result
for minimal elliptic operators with smooth coefficients,
see \cite{Ber65,Gru08,LioMag72}. Besides we show that the
Dirichlet and Neumann realizations defined by
\bead
\dom(A^D) & := &  \{f\in W^{2,2}_T(\R_+,\cH):\ f(0)=0\},\\
\dom(A^N) & := & \{f\in W^{2,2}_T(\R_+,\cH):\ f'(0)=0\}
\eead
are self-adjoint, cf. Proposition \ref{prop6.8}. This statement
is similar to that of the regularity of Dirichlet and Neumann
realizations in elliptic theory (cf. \cite{Ber65,Gru08,LioMag72}).
It looks surprising, that these regularity statements
were not obtained in previous papers even in the case of finite intervals.

Moreover, we  show that the realizations $A^D$ and $A^N$ are absolutely
continuous and unitarily equivalent for any $T$. We note
that these results can easily be obtained using the tensor product
structure of $A^D$ and $A^N$, see Appendix A.2.
However, the method fails if the special tensor product structure is missing.
We investigate the spectral properties of arbitrary self-adjoint realizations of
$\cA$ by investigating the corresponding Weyl functions.

We point out that the results substantially differ from those for
Dirichlet and Neumann extensions $A^D_I$ and $A^N_I$ of $\cA$ on a
finite interval $I$. In the later case  \emph{the spectral
properties of $A^D_I$ and $A^N_I$ strongly correlate with those
of} $T$, cf. Appendix A.1. In
particular, we show that, in contrast  to the case of a finite
interval, for any $T=T^*\ge 0$ none of the realizations of
$\cA$ on the semi-axis is pure point, purely singular or discrete.
Moreover, we show that for any $T\ge 0$  the Dirichlet and the Neumann
realizations  $A^D$ and $A^N$ are $ac$-minimal in the following sense.
\bd[{\cite[Definition 3.5, Definition 5.1]{MN2011}}]\la{III.5}
{\em
Let $A$ be a closed symmetric operator and let $A_0$ be a
self-adjoint extension of $A$.
\item[\;\;(i)]  We say that $A_0$ is $ac$-minimal if for any
self-adjoint extension $\wt A$ of $A$ the absolutely continuous part
$A^{ac}_0$ is unitarily equivalent to a part of $\wt A$.
\item[\;\;(ii)] Let $\gs_0 := \gs_{ac}(A_0)$. We say that $A_0$ is
strictly $ac$-minimal if for any self-adjoint extension $\wt A$ of
$A$ the part  $\wt A^{ac}E_{\wt A}(\gs_0)$ of $\wt A$ is unitarily
equivalent to the absolutely continuous part $A^{ac}_0$ of $A_0$.
}
\ed
One of our main results, which follows from Theorem \ref{VI.9},
Theorem \ref{VI.10a} and Corollary \ref{cor6.12},
can be summarized as follows:
\bt\label{th0.1}
Let $T$ be a non-negative self-adjoint operator in the infinite
dimensional Hilbert space $\cH$ with
$t_0 = \inf\gs(T)$ and $t_1 = \inf\gs_{\ess}(T)$.
Further, let $\wt A$ be a
self-adjoint realization of $\cA$. Then the following holds:
\item[\rm \;\;(i)]
The Dirichlet and  the Neumann realizations
$A^D$ and $A^N$ of $\cA$ are unitarily equivalent, absolutely continuous and
$\gs(A^D) = \gs_{ac}(A^D) = \gs(A^N) =
\gs_{ac}(A^N) = [t_0, \infty)$.

\item[\rm \;\;(ii)]
The Dirichlet, Neumann and Krein
realizations $A^D$, $A^N$ and $A^K$ of $\cA$ are
$ac$-minimal.

\item[\rm \;\;(iii)]
These realizations  are strictly $ac$-minimal if and only if $t_0 = t_1.$

\item[\rm \;\;(iv)]
If one of the following conditions
\bed
(\wt A - i)^{-1}- (A^D - i)^{-1}\in \gotS_\infty(\gotH)\quad
\text{or}\quad
  (\wt A - i)^{-1}- (A^K - i)^{-1}\in \gotS_\infty(\gotH)
\eed
is satisfied, then the absolutely continuous part $\wt A^{ac}$
of  $\wt A$  is unitarily equivalent to the Dirichlet realization
$A^D$.

\item[\rm \;\;(v)] If $t_0 = t_1$, then the absolutely continuous part $\wt A^{ac}$
of  $\wt A$   is unitarily equivalent to the Dirichlet realization
$A^D$ provided that
\bed
(\wt A - i)^{-1}- (A^N - i)^{-1}\in \gotS_\infty(\gotH).
\eed
\et

At first glance it seems that the $ac$-minimality of $A^D$ contradicts the
classical Weyl-v.~Neumann theorem, cf. \cite[Theorem X.2.1]{Ka76},  which guarantees the
existence   of a Hilbert-Schmidt perturbation $C=C^*$
such that the spectrum $\sigma(A^D + C)$ of the
perturbed operator $A^D + C$  is pure point. But, in fact,
Theorem \ref{th0.1} presents an explicit example showing  that the
analog of the Weyl-v.Neumann theorem  does not hold for
non-additive classes of perturbations. {Indeed, Theorem \ref{th0.1} shows
that for the class of self-adjoint extensions of $A$
the absolutely continuous part can never be eliminated. Moreover, if
$(\wt A - i)^{-1}- (A^D - i)^{-1}$ is compact, then even unitary
equivalence holds.}

We apply Theorem \ref{th0.1} and other abstract results to
Schr\"odinger operators
\bed
\cL := -\frac{\partial^2}{\partial t^2}
-\sum^n_{j=1}\frac{\partial^2}{\partial x^2} + q(x) =
-\frac{\partial^2}{\partial t^2}  -\Delta_x + q,
\qquad (t,x) \in \R_+\times\R^n,
\eed
considered in the half-space $\R^{n+1}_+ = \R_+\times \R^n$, $n\in
\N$. Here $q$ is a bounded non-negative  potential, $q=\bar q\in
L^\infty(\R^n),\ q\ge 0$. In this case the minimal elliptic
operator $L:= L_{\min}$  generated in $L^2(\R^{n+1}_+)$ by the
differential expression $\cL$ can be identified with the minimal
operator $A = A_{\min}$ generated in   $\gotH = L^2(\R_+,\cH)$,
$\cH:= L^2(\R^n)$, by the  differential expression \eqref{8.10}
with $T= -\Delta_x + q = T^*$. Therefore and due to the regularity
theorem (see \cite{Gru08,LioMag72}) the Dirichlet $L^D$ and the Neumann $L^N$
realizations  of the elliptic expression $\cL$ are identified,
respectively, with the realizations $A^D$ and $A^N$ of the expression $\cA$.
Moreover, the Krein realization $L^K$ of $\cL$ is identical with $A^K$. This leads to
statements on realizations of $\cL$ which are similar to those of Theorem \ref{th0.1}.
In fact, one has only to replace $A$ by $L$ in Theorem \ref{th0.1}. In addition,
if the condition
\begin{equation}\label{0.33}
\lim_{|x|\to\infty}\int_{|x-y|\le 1} q(y) dy = 0
\end{equation}
is satisfied, then $L^D$ and $L^N$  are absolutely continuous and strictly
$ac$-minimal. In particular, $\gs(L^D) = \gs_{ac}(L^D) = \gs(L^N)
= \gs_{ac}(L^N) = [0,\infty)$.

To prove {Theorem \ref{th0.1}} we consider the minimal symmetric
operator $A$ associated with the differential expression $\cA$
in the framework of extension theory, more precisely, in the framework
of boundary triplets intensively developed during the last three decades,
see for instance \cite{DM91,DM95,GG91} or \cite{BGPankr2008}
and references therein. The key role  in this theory plays  the so-called abstract
Weyl function introduced and investigated in
\cite{DM87,DM91,DM95}. Moreover, the proofs invoke techniques
elaborated in \cite{ABMN05,BMN02} and our recent publication \cite{MN2011}.

Namely, the proofs of unitary equivalence are based on some
statements from \cite{MN2011},
which allow to compute the  spectral multiplicity function $N_{\wt
A^{ac}}(\cdot)$ of the $ac$-part $\wt A^{ac}$ of an extension $\wt
A = \wt A^*$ in terms of boundary values of the
Weyl functions at the real axis, {cf. Proposition \ref{III.8} and
Corollary \ref{IV.9}.}

We construct a special boundary triplet for the operator $A^*$ (in
the case of unbounded $T=T^*\ge0$) representing $A$ as a direct
sum of minimal Sturm-Liouville operators $S_n$ with bounded
operator potentials {$T_n := TE_T([n-1,n))$, $n \in \N$,} where $E_T(\cdot)$
is the spectral measure of $T$. The corresponding Weyl function
$M(\cdot)$ has weak boundary values
\be\label{0.2A}
M(\gl) := M(\gl+i0)= \wlim_{y\downarrow 0}M(\gl + iy) \quad\text{for
a.e.} \quad  \gl \in \R.
\ee
This boundary triplet differs from that used in \cite[Section
9]{DM91}. {It is  more suitable  for the investigation of the
$ac$-spectrum of realizations of $\cA$
than that one of \cite[Section 9]{DM91}}.
Due to the {property \eqref{0.2A} the statement (iv) of Theorem
\ref{th0.1} follows immediately from
our recent result \cite[Theorem 1.1]{MN2011}). We note that this is
more than one can expect when applying the classical Kato-Rosenblum theorem
\cite{Ka76,Ros57}. Indeed, in accordance with its generalization by
Kuroda \cite{Kur59,Kur60}, Birman \cite{Bir63a} and
Birman and Krein \cite{BirKrei62} it is required
that the resolvent differences in (iv) and (v) of Theorem \ref{th0.1}
belong to the trace class ideal and not to the compact one as
actually assumed}. We note also that although the
limit \eqref{0.2A} does not exist for the Weyl
function of the Neumann realization $A^N$ the conclusion
(iv) of Theorem \ref{th0.1} still remains valid, cf. Theorem \ref{th0.1}(v).

The paper is organized as follows. In Section 2 we give a short
introduction into the theory of boundary triplets and  the
corresponding Weyl functions. {We recall here some statements
on spectral multiplicity functions and
the main theorem from \cite{MN2011} used in the following}.

In Section 3 we obtain some new results on symmetric operators $S
:= \bigoplus^\infty_{n=1} S_n$ being  an infinite direct sum of
closed symmetric operators $S_n$ with equal deficiency indices.
First,  let $\Pi_n= \{\cH_n, \gG_{0n}, \gG_{1n}\}$ be a boundary
triplet for $S^*_n,$  $n \in \N$. In general, the direct sum $\Pi
= \bigoplus^\infty_{n=1}\Pi_n$ is not a boundary triplet for $S^*
= \bigoplus^\infty_{n=1} S^*_n$, cf. \cite{Koch79}. Nevertheless,
we show, cf. Theorem \ref{VI.3},  that each  boundary triplet $\Pi_n$ can  slightly be
modified such that the new sequence $\wt\Pi_n = \{\cH_n,
\wt\gG_{0n}, \wt\gG_{1n}\}$ of boundary triplets possess the
following properties:
\begin{enumerate}

\item[(i)]
the direct sum
\bed
\wt\Pi = \bigoplus^\infty_{n=1}\wt\Pi_n =  \{\cH, \wt\gG_{0},
\wt\gG_{1}\}, \quad
 \cH:= \bigoplus^\infty_{n=1}\cH_n, \quad
\wt\gG_{j}:= \bigoplus^\infty_{n=1}\wt\gG_{jn}, \quad j\in\{0.1\},
\eed
is already a boundary triplet for $S^*$;

\item[(ii)]
the extension $\wt S_{0}:= S^*\upharpoonright\ker \wt\gG_{0}$
satisfies $\wt S_{0} = \bigoplus^\infty_{n=1}\wt S_{0n}$ where
\bed
\wt S_{0n}:= S^*_n\upharpoonright\ker \wt\gG_{0n} =  S^*_n\upharpoonright
\ker\gG_{0n} =: S_{0n}, \quad n\in \N.
\eed
\end{enumerate}
Moreover, the Weyl function $\wt M(\cdot)$ corresponding
to the triplet $\wt\Pi$ is block-diagonal, that is,
$\wt M(\cdot) = \bigoplus^\infty_{n=1}\wt M_n(\cdot)$ where $\wt M_n(\cdot)$ is
the Weyl function corresponding to the triplet $\wt\Pi_n$, $n\in
\N$. This result plays an important role in the sequel. In
particular, we show that the self-adjoint extension $S_0 =
\bigoplus^\infty_{n=1}S_{0n}$ is $ac$-minimal provided
{that the deficiency indices $n_\pm(S_n)$ are equal and finite}.
We also prove in this section that if $S_n\ge 0, n\in \N,$ then the
Friedrichs and Krein extensions $S^F$ and  $S^K$ of $S :=
\bigoplus^\infty_{n=1} S_n$, respectively, are the direct sums of Friedrichs
and Krein extensions of the summands $S_n$, i.e., $S^F :=
\bigoplus^\infty_{n=1} S_n^F$ and $S^K := \bigoplus^\infty_{n=1}
S_n^K$, cf. Corollary \ref{cor5.5}. In a recent paper \cite{KosMal10}
Theorem \ref{VI.3} has been applied  to Schr\"odinger operators with local
point interactions.

In Section 4 we consider Sturm-Liouville operators with bounded operator potentials. In
this case it is easy to construct a boundary triplet for $A^*$. We
prove here Theorem \ref{th0.1} in the case $T\in [\cH]$ and
establish some additional properties of  Krein's realization as
well as other realizations.

In Section 5 we extend the results to the case of Sturm-Liouville
operators with unbounded non-negative operator potentials.
We construct here a boundary triplet for $A^*$ using results of
both Sections 3 and 4 and compute the (block-diagonal) Weyl
function. Based on this construction we prove Theorem \ref{th0.1}
for unbounded $T$  and establish some additional properties of
Dirichlet, Neumann and  other realizations as well. In particular, we prove
here the regularity results mentioned above. Finally,
we apply the abstract results to the elliptic partial
differential expression $\cL$ in the half-space.

In the Appendix we present some results on
realizations of $\cA$ admitting separation of variables, i.e.,
having a certain tensor product structure.

The main results of the paper have been announced (without proofs)
in \cite{MN2010}, a preliminary version has been published as a
preprint \cite{MN2009}. Since the results of the paper are obvious if
$\dim(\cH) < \infty$ we consider the case
when $\dim(\cH) = \infty$.

{\bf Notations}  In the following we consider only separable
Hilbert spaces which are denoted by $\mathfrak H$, $\cH$ etc.
A closed linear relation in $\cH$ is a closed subspace of $\cH \oplus
\cH$. The set of all closed linear relations in $\cH$ is denoted by
$\wt \cC(\cH)$. A graph $\graph(B)$ of a closed linear operator $B$ belongs
to $\wt \cC(\cH)$. The symbols $\cC(\cH_1, \cH_2)$ and $[\mathfrak H_1,
\mathfrak H_2]$ stand for the sets of closed and
bounded linear operators from
$\mathfrak H_1$ to $\mathfrak H_2,$ respectively.
We set $\cC(\cH):= \cC(\cH, \cH)$ and $[\mathfrak H] :=
[\mathfrak H, \mathfrak H]$.
We regard $\cC(\cH)$ as a subset of $\wt\cC(\cH)$ identifying an
operator $B$ with its graph $\graph(B)$.

The Schatten-v.~Neumann ideals of compact operators
are denoted by $\gotS_p(\gotH)$, $p \in [1,\infty]$, where
$\gotS_1(\gotH)$, $\gotS_2(\gotH)$ and $\gotS_\infty(\gotH)$ are the
ideals of trace, Hilbert-Schmidt and compact operators, respectively.

The symbols $\dom(T)$, $\ran(T)$, $\varrho(T)$ and $\sigma(T)$ stand for the domain,
the range,  the resolvent set and the spectrum of an operator
$T\in \cC(\cH),$ respectively; $T^{ac}$ and
$\sigma_{ac}(T)$ stand for the absolutely continuous part and the
absolutely continuous spectrum of a self-adjoint operator $T=T^*$.

\section{Preliminaries}

\subsection{Boundary triplets and proper extensions}

In this section we briefly recall basic facts on  boundary triplets and their
Weyl functions, cf. \cite{DM87,DM91,DM95,GG91}.

Let $A$ be a densely defined closed symmetric operator in the
separable Hilbert space $\gH$ with equal deficiency indices
$n_\pm(A)=\dim(\ker(A^*\mp i)) \leq \infty$.
\begin{definition}[\cite{GG91}]\la{t1.9}
{\em
A triplet $\Pi = \{\cH, \Gamma_0,
\Gamma_1\}$, where $\cH$ is an auxiliary Hilbert space and
$\Gamma_0,\Gamma_1:\  \dom(A^*)\rightarrow \cH$ are linear
mappings,  is called an boundary triplet for $A^*$ if
the "abstract Green's identity"
\begin{equation}\la{2.10A}
(A^*f,g) - (f,A^*g) = (\gG_1f,\gG_0g)_{\cH} -
(\gG_0f,\gG_1g)_{\cH}, \qquad f,g\in\dom(A^*),
\end{equation}
holds and the mapping $\gG := (\Gamma_0,\Gamma_1):  \dom(A^*)
\rightarrow \cH \oplus \cH$ is surjective.
}
\end{definition}
\begin{definition}[\cite{GG91}]\la{def2.2}
{\em
A closed extension $A'$ of $A$ is called a proper extension,
in short $A'\in \Ext_A,$ if $A\subset A' \subset A^*$.
Two proper extensions $A', A''$ are called disjoint if
$\dom( A')\cap \dom( A'') = \dom( A)$ and  transversal if in
addition $\dom( A') + \dom( A'') = \dom( A^*).$
}
\end{definition}

Clearly, any  self-adjoint extension $\wt A= {\wt A}^*$ is proper,
$\wt A\in \Ext_A$. A boundary triplet $\Pi=\{\cH,\gG_0,\gG_1\}$ for $A^*$ exists
whenever  $n_+(A) = n_-(A)$. Moreover, the relations $n_\pm(A) =
\dim(\cH)$ and $\ker(\Gamma_0) \cap \ker(\Gamma_1)=\dom(A)$ are
valid. In addition one has $\Gamma_0, \Gamma_1\in [{\mathfrak H}_+,\cH]$
where ${\mathfrak H}_+$ denotes the Hilbert space { obtained
by equipping} $\dom(A^*)$ with the graph norm of $A^*$.

Using the concept of { boundary triplets} one can parameterize all
proper, in particular, self-adjoint extensions of $A$. For this
purpose we denote by $\wt\cC(\cH)$ the set  of closed linear
relations in $\cH$, that is, the set of all closed linear subspaces
of $\cH\oplus\cH$.  A linear relation $\gT$ is called {\it symmetric} if
$\gT\subset\gT^*$ and self-adjoint if $\gT=\gT^*$ where
$\gT^*$ is the adjoint relation.
For the definition of the inverse and the resolvent set of a linear
relation $\gT$ we refer to \cite{DS87}.
\begin{proposition}\label{prop2.1}
Let  $\Pi=\{\cH,\gG_0,\gG_1\}$  be a
boundary triplet for  $A^*$.  Then the mapping
\be\label{bij}
\Ext_A \ni \wt A \to  \Gamma \dom(\wt A)
=\{\{\Gamma_0 f,\Gamma_1f \} : \  f\in \dom(\wt A) \} =:
\gT \in \wt\cC(\cH)
\ee
establishes  a bijective correspondence between the sets $\Ext_A$
and  $\wt\cC(\cH)$. We put $A_\gT :=\wt A$ where
$\gT$ is defined by \eqref{bij}. Moreover, the following holds:

\item[\;\;\rm (i)] $A_\gT= A_\gT^*$ if and only if $\gT=
\gT^*$;

\item[\;\;\rm (ii)] The extensions $A_\gT$ and $A_0$ are disjoint
if and only if there is an operator $B \in \cC(\cH)$ such that
$\graph(B) = \gT$. In this case \eqref{bij} takes the form
\bed
A_\gT = A^*\!\upharpoonright\ker(\gG_1- B\gG_0);
\eed

\item[\;\;\rm (iii)] The extensions $A_\gT$ and $A_0$ are
transversal if and only if $A_\gT$ and $A_0$ are disjoint and $\gT =
\graph(B)$ where $B$ is bounded.
\end{proposition}

With any boundary triplet $\Pi$ one associates  two special extensions
$A_j:=A^*\!\upharpoonright\ker(\gG_j), \ j\in \{0,1\}$, which are
self-adjoint in view of Proposition \ref{prop2.1}. Indeed, we have
$A_j:=A^*\!\upharpoonright\ker(\gG_j) = A_{\gT_j},\ j\in
\{0,1\}$, where $\gT_0:= \{0\} \times \cH$ and $\gT_1 := \cH \times
\{0\}$. Hence $A_j= A_j^*$ since $\gT_j = \gT^*_j$.
In the sequel the extension $A_0$ is usually regarded as a reference
self-adjoint extension.

{Moreover, if $\gT$ is the graph of a closed
operator $B$, i.e. $\gT = \graph(B)$, then the operator $A_\gT$ is
denoted by $A_B$.}

Conversely, for any extension $A_0=A_0^*\in \Ext_A$ there exists
a boundary triplet $\Pi=\{\cH,\gG_0,\gG_1\}$ for $A^*$
such that $A_0:=A^*\!\upharpoonright\ker(\gG_0)$.

\subsection{Weyl functions and $\gamma$-fields}\label{weylsec}

It is well known that Weyl functions are an important tool in the
direct and inverse spectral theory of singular Sturm-Liouville
operators. In \cite{DM87,DM91,DM95} the concept of Weyl function
was generalized to the case of an arbitrary symmetric operator $A$
with  $n_+(A) = n_-(A).$ Following \cite{DM87,DM91,DM95} we recall
basic facts on Weyl functions and $\gamma$-fields associated with
a boundary triplet $\Pi$.
\bd[{\cite{DM87,DM91}}]\label{Weylfunc}
{\em
 Let $\Pi=\{\cH,\gG_0,\gG_1\}$ be a boundary triplet  for $A^*.$
The functions $\gamma(\cdot):
\varrho(A_0)\rightarrow [\cH,\gH]$ and  $M(\cdot):
\varrho(A_0)\rightarrow  [\cH]$ defined by
\begin{equation}\label{2.3A}
\gamma(z):=\bigl(\Gamma_0\!\upharpoonright\mathfrak N_z\bigr)^{-1}
\qquad\text{and}\qquad M(z):=\Gamma_1\gamma(z), \qquad
z\in\varrho(A_0),
      \end{equation}
are called the {\em $\gamma$-field} and the {\em Weyl function},
respectively, corresponding to  $\Pi.$
}
\ed

It follows from the identity  $\dom(A^*)=\ker(\Gamma_0)\dot +
\gotN_z$, $z\in\varrho(A_0)$, where  $A_0 =
A^*\!\upharpoonright\ker(\gG_0)$, and ${\mathfrak N}_z :=\ker(A^* -
z)$, that the $\gamma$-field $\gamma(\cdot)$ is well defined and
takes values in $[\cH,\gH]$. Since $\gG_1 \in [\gH_+, \cH]$, it
follows from \eqref{2.3A} that  $M(\cdot)$ is well defined too and
takes values in  $[\cH]$. Moreover, both $\gamma(\cdot)$ and
$M(\cdot)$ are holomorphic on $\varrho(A_0)$. It turns out
than the Weyl function $M(\cdot)$ is  in fact
a $R_{\cH}$-function (Nevanlinna or Herglotz function), that is, $M(\cdot)$ is a
$[\cH]$-valued holomorphic function on $\C\backslash \R$  satisfying
\bed
M(z)=M(\overline z)^*\qquad\text{and}\qquad
\frac{\IM(M(z))}{\IM(z)}\geq 0, \qquad z\in\C\backslash\R,
 \eed
which in addition satisfies the condition
$0\in \varrho(\IM(M(z))),\  z\in\C\backslash\R$.

If $A$ is a simple symmetric operator, then  the Weyl function
$M(\cdot)$ determines the pair $\{A,A_0\}$ uniquely up to unitary
equivalence (see \cite{DM95,KL71}). Therefore $M(\cdot)$ contains
(implicitly)  full information on spectral properties of $A_0$.
We recall that a symmetric operator is
said to be {\it simple} if there is no non-trivial subspace which
reduces it to a self-adjoint operator.

For a fixed $A_0 = A_0^*$ extension of $A$ the boundary triplet
$\Pi =\{\cH, \gG_0,\gG_1\}$ satisfying  $\dom(A_0)=\ker
(\Gamma_0)$ is not unique.  If  $\wt \Pi =\{\wt \cH, \wt
\gG_0,\wt \gG_1\}$ is  another boundary triplet  for $A^*$
satisfying  $\ker(\gG_0) = \ker(\wt \gG_0)$,  then the
corresponding Weyl functions $M(\cdot)$ and $\wt M(\cdot)$ are
related by
\begin{equation}\label{2.11}
\wt M(z) = R^* M(z) R + R_0,
\end{equation}
where $R_0= R^*_0 \in [\wt \cH]$ and $R \in
[\wt \cH,\cH]$ is boundedly invertible.

\subsection{Krein type formula for resolvents and resolvent comparability}

With  any  boundary triplet   $\Pi=\{\cH,\gG_0,\gG_1\}$  for $A^*$
and any proper (not necessarily  self-adjoint)  extension $A_{\gT}\in
\Ext_A$ it is naturally  associated  the following (unique) Krein
type formula (cf. \cite{DM87,DM91,DM95})
\be\label{2.30}
(A_\gT - z)^{-1} - (A_0 - z)^{-1} = \gamma(z) (\gT -
M(z))^{-1} \gamma({\overline z})^*, \quad z\in \varrho(A_0)\cap
\varrho(A_\gT).
\ee
Formula \eqref{2.30} is a generalization of the known Krein
formula for resolvents.
We note also, that all objects in
\eqref{2.30} are expressed in terms of the boundary triplet $\Pi$
(cf. \cite{DM87,DM91,DM95}). The following result is deduced from  formula \eqref{2.30} (cf.
\cite[Theorem 2]{DM91}).
\bp\label{prop2.9}
Let $\Pi=\{\cH,\gG_0,\gG_1\}$  be a boundary triplet for $A^*$,
$\gT_i = \gT_i^* \in \wt\cC(\cH), \ i\in \{1,2\}$.  Then for
any Schatten-v.~Neumann ideal ${\gotS}_p$, $p \in (0,\infty]$, and any $z \in
\C\setminus\R$ the following equivalence holds
\bed
(A_{\gT_1}-z)^{-1} - (A_{\gT_2}-z)^{-1}\in{\mathfrak
S}_p(\gH)
\Longleftrightarrow  \bigl(\gT_1 - z\bigr)^{-1}-
\bigl(\gT_2 - z \bigr)^{-1}\in{\mathfrak S}_p(\cH)
\eed
In particular,  $(A_{{\gT}_1} - z)^{-1} - (A_0 - z)^{-1} \in
{\mathfrak S}_p(\gH) \Longleftrightarrow \bigl(\gT_1 - i\bigr)^{-1} \in {\mathfrak
S}_p(\cH)$.

If in addition $\gT_1, \gT_2\in[\cH]$, then for any $p \in
(0, \infty]$ the equivalence holds
\bed
(A_{\gT_1}-z)^{-1} -
(A_{\gT_2}-z)^{-1}\in{\mathfrak S}_p(\gH) \Longleftrightarrow
\gT_1 - \gT_2 \in{\mathfrak S}_p(\cH).
\eed
\end{proposition}

\subsection{Spectral multiplicity function and unitary equivalence}

Let as above  $A$ be a densely defined \emph{simple} closed
symmetric operator in $\gH$ and let $\Pi = \{\cH,\gG_0,\gG_1\}$ be
a boundary triplet for $A^*$,  $M(\cdot)$ the corresponding Weyl
function $M(\cdot)$ and $A_0 = A^*\!\upharpoonright\ker(\gG_0)=
 A_0^*$.

In our recent publication  \cite{MN2011} using some results from \cite{MM03}
we expressed the spectral multiplicity function
$N_{A_0^{ac}}(\cdot)$ of $A_0^{ac}$  by means of the limit values
of the Weyl function $M(\cdot)$. In general, the limit $M(t) :=
\slim_{y\downarrow 0}M(t+iy)$, $t \in \R$, does not exist.
However, for any  $D\in\gotS_2(\cH)$ 
satisfying $\ker(D) = \ker(D^*) = \{0\}$ the ``sandwiched'' Weyl
function,
\bed
M^D(z) := D^*M(z)D, \quad z \in \C_\pm,
\eed
admits limit values  $M^D(t) := \slim_{y\downarrow 0}M^D(t+iy)$
for a.e. $t \in \R$,  even in $\gotS_2$-norm (cf.
\cite{BirEnt67}, \cite{Gin66}). We set
\bed
d_{M^D}(t)  := \dim(\ran(\im(M^D(t)))),
\eed
which is well-defined for a.e. $t \in \R$. The function
$d_{M^D}(\cdot)$ is Lebesgue measurable and takes values in the
set of extended natural numbers $\{0\} \cup \N \cup \{\infty\} =
\{0,1,2,\ldots,\infty\}$. The set $\supp_{d_{M^D}} := \{t \in \R:
d_{M^D}(t) > 0\}$ is called the support of $d_{M^D}(\cdot)$ and
is, of course, a Lebesgue measurable set of $\R$. If the limit
$M(t) := \slim_{y\downarrow 0}M(t+iy)$ exists for a.e. $t \in \R$,
then we set $d_M(t)  := \dim(\ran(\im(M(t))))$.

To state the next result we introduce the notion of  the
absolutely continuous closure $\cl_{ac}(\delta)$ of a Borel subset $\delta\subset \R$
(see for definition \cite[Appendix]{MN2011} as well as
\cite{BMN02,GMZ08}). The use of this notion  for the investigation of
the $ac$-spectrum of Schr\"odinger operators etc. see
the recent publication \cite{Geszin2009}.
\bp[{\cite[Proposition 3.2]{MN2011}}]\la{III.8}
Let $A$ be as above 
and let $\gP = \{\cH,\gG_0,\gG_1\}$ be a boundary triplet for
$A^*$, $M(\cdot)$ the corresponding Weyl function.  If $D$ is a
Hilbert-Schmidt operator such that $\ker(D) = \ker(D^*) = \{0\}$,
then $N_{{A^{ac}_0}}(t) = d_{M^D}(t)$ for a.e. $t \in \R$ and
$\gs_{ac}(A_0) = \cl_{ac}(\supp(d_{M^D}))$.

If, in addition, the limit $M(t) := \slim_{y\downarrow 0}M(t+iy)$ exists
for a.e. $t \in \R$, then $N_{{A^{ac}_0}}(t) = d_M(t)$ for a.e. $t
\in \R$ and $\gs_{ac}(A_0) = \cl_{ac}(\supp(d_{M}))$.
\end{proposition}
If $\wt A = \wt A^*\in \Ext_A$ and is disjoint with $A_0$, then by
Proposition \ref{prop2.1}(ii) there is a self-adjoint operator $B$
acting in $\cH$ such that $\wt A = A_B :=
A^*\upharpoonright\ker(\gG_1 - B\gG_0)$. In this case the
multiplicity function  $N_{{A^{ac}_B}}(\cdot)$ is  expressed by
means of
the generalized Weyl function $M_B(\cdot)$ of $\wt A= A_B$ defined
by
\be\la{2.5}
M_B(z) := (B - M(z))^{-1}, \quad z \in \C_\pm,
\ee
\bc[{\cite[Corollary 3.3]{MN2011}}]\la{IV.9}
Let $A$, $\Pi$, $M(\cdot)$ and $D$  be as in Proposition
\ref{III.8} and let $B = B^*\in \cC(\cH).$ Then $N_{A^{ac}_B}(t) =
d_{M_B^D}(t)$ for a.e. $t\in \R$ and $\gs_{ac}(A_B) =
\cl_{ac}(\supp(d_{M_B^D}))$.

If, in addition, the limit
$M_B(t) := \slim_{y\downarrow 0}M_B(t + iy)$ exists for a.e. $t \in \R$, then
$N_{A^{ac}_B}(t) = d_{M_B}(t)$ for a.e. $t \in \R$ and
$\gs_{ac}(A_B) = \cl_{ac}(\supp(d_{M_B}))$.
\ec
Finally, we can retranslate the unitary equivalence of $ac$-parts
of two self-adjoint extensions in terms of the limit values of the
Weyl functions.
\bt[{\cite[Theorem 3.4]{MN2011}}]\la{IV.10}
Let $A$, $\Pi$, $M(\cdot)$ and $D$  be as in Proposition
\ref{III.8} and $B = B^*\in \cC(\cH).$  Let also  $E_{A_B}(\cdot)$
and $E_{A_0}(\cdot)$ be the spectral measures of  $A_B=A_B^*$ and
$A_0$, respectively. If $\gd$ is a Borel
subset of $\R$, then

\item[{\rm\;\;(i)}] $A_0E^{ac}_{A_0}(\gd)$ is unitarily equivalent
to a part of $A_BE^{ac}_{A_B}(\gd)$ if and only if
$d_{M^{ D}}(t) \le d_{M_B^{ D}}(t)$ for
a.e. $t \in \gd$;

\item[{\rm\;\;(ii)}] $A_0E^{ac}_{A_0}(\gd)$
and $A_BE^{ac}_{A_B}(\gd)$  are unitarily equivalent if and only if
$d_{M^D}(t) = d_{M_B^D}(t)$ for a.e. $t \in \gd.$
\et
Theorem \ref{IV.10} reduces the problem of unitary
equivalence of $ac$-parts  of certain self-adjoint extensions of
$A$ to the computation of the functions $d_{M^{D}}(\cdot)$ and
$d_{M_B^{D}}(\cdot)$. If $\gd = \R$, then the absolutely continuous
part $A^{ac}_0$ is unitarily equivalent to $\wt A^{ac} = A^{ac}_B$ if
and only if $d_{M^D}(t) = d_{M^D_B}(t)$ for a.e. $t \in \R$.

If $M(\cdot)$ is the Weyl function of a boundary triplet $\Pi$,
then we introduce the maximal normal function
\bed
m^+(t) := \sup_{y\in (0,1]}\left\|M(t+iy)\right\|,\quad  t \in \R.
\eed
\bt[{\cite[Theorem 4.3, Corollary 4.6]{MN2011}}]\la{V.5}
Let $A$, $\Pi$, $M(\cdot)$ and $D$  be as in Proposition
\ref{III.8}. Let $\wt A = \wt A^*\in \Ext_A$  and
$A_0 := A^*\upharpoonright\ker(\gG_0)$. Assume also that there is a Borel
subset $\gd$ of $\R$ such that the maximal
normal function $m^+(t)$ is finite for a.e. $t \in \gd$ and the
condition
\be\la{0.2}
({\wt A}-i)^{-1}-(A_0-i)^{-1}\in{\gotS}_{\infty}(\gH),
\ee
is satisfied.  Then the $ac$-parts $\wt A^{ac}E_{\wt A}(\gd)$ of
$\wt AE_{\wt A}(\gd)$  and $A_0E_{A_0}(\gd)$, respectively,  are
unitarily equivalent. In particular, if $m^+(t)$ is finite for
a.e. $t \in \R$, then absolutely continuous parts $\wt A^{ac}$ and
$A^{ac}_0$ are unitarily equivalent.
\et

One easily verifies that $m^+(t)< \infty$ for a.e. $t \in \gd$ if
and only if limit \eqref{0.2A} exists for a.e. $t \in \gd$. Thus,
condition $m^+(t) < \infty$ for a.e. $t \in \gd$ in Theorem
\ref{V.5} can be replaced by the assumption that the limit
\eqref{0.2A} exists for a.e. $t \in \gd$, cf. \cite[Theorem
1.1]{MN2011}.

However, the function $m^+(\cdot)$ depends on  the chosen  boundary triplet.
In \cite{MN2009}-\cite{MN2011} we introduced the invariant maximal
normal function $\gotm^+(\cdot)$ defined by
\be\label{4.12A}
\gotm^+(t) := \sup_{y\in(0,1]}
\left\|
\frac{1}{\sqrt{\im(M(i))}}\left(M(t+iy) -
\re(M(i))\right)\frac{1}{\sqrt{\im(M(i))}}
\right\|,
\ee
$t \in \R$.
It follows from \eqref{2.11} that the invariant maximal normal
functions
for  two boundary triplets $\Pi = \{\cH,\gG_0,\gG_1\}$  and $\wt
\Pi = \{\wt H,\wt \gG_0,\wt \gG_1\}$ for $A^*$  coincide whenever
$A^*\upharpoonright\ker(\gG_0) = A^*\upharpoonright\ker(\wt
\gG_0)$.
Clearly, $\gotm^+(t) < \infty$ if and only if $m^+(t) < \infty$
for any $t \in \R$. However, the invariant maximal normal function
is more convenient in applications. We demonstrate this fact in
the next section applying this concept  to infinite direct sums of
symmetric operators.

\section{Direct sums of symmetric operators}

\subsection{Boundary triplets for direct sums}

Let ${ S_n}$ be a closed densely defined symmetric {
operators} in $\gotH_n,$ \ $n_+({ S_n})=n_-( S_n),$ and { let
$\gP_n = \{\cH_n,\gG_{0n},\gG_{1n}\}$  be} a boundary triplet for
$S_n^*,\ n\in \N.$ Let
\be\la{6.5a}
A := \bigoplus^\infty_{n=1}{ S_n}, \quad \dom(A) :=
\bigoplus^\infty_{n=1}\dom({ S_n}).
\ee
Clearly, $A$ is a closed densely defined symmetric operator  in
the Hilbert space  $\gotH := \bigoplus^\infty_{n=1}\gotH_n$ with
$n_\pm(A)= \infty$. Obviously, we have
\be\la{6.5b}
A^* = \bigoplus^\infty_{n=1}{ S}^*_n, \quad \dom(A^*) =
\bigoplus^\infty_{n=1}\dom({ S}^*_n).
\ee
Let us consider the direct sum
$\gP:=\bigoplus_{n=1}^{\infty} \gP_n=: \{\cH,\gG_{0},\gG_{1}\}$ of
 boundary triplets defined by
\be\la{6.5}
\cH := \bigoplus^\infty_{n=1}\cH_n, \quad \gG_0 :=
\bigoplus^\infty_{n=1}\gG_{0n} \quad \mbox{and} \quad \gG_1 :=
\bigoplus^\infty_{n=1}\gG_{1n}.
\ee
We note that the Green's identity
\bed (S_n^*f_n, g_n) - (f_n,S_n^*g_n) =
(\gG_{1n}f_n,\gG_{0n}g_n)_{\cH_n} - (\gG_{0n}f_n,
\gG_{1n}g_n)_{\cH_n}, \eed
$f_n, g_n \in \dom( S_n^*)$, holds for every $S^*_n$, $n \in
\N$. This yields that the Green's identity \eqref{2.10A} holds for $A_*:=
A^*\upharpoonright \dom (\gG)$, $\dom(\gG) := \dom (\gG_0) \cap
\dom (\gG_1) \subseteq \dom (A^*)$, that is, for
$f=\bigoplus_{n=1}^{\infty}f_n$, $g = \bigoplus_{n=1}^\infty g_n \in
\dom (\gG)$ we have
\begin{equation}\label{3.1AA}
(A_*f, g) - (f, A_*g) = (\gG_1f, \gG_0 g)_{\cH} -
(\gG_0 f, \gG_1 g)_{\cH}, \qquad f,g\in\dom(\gG),
\end{equation}
where $A^*$ and $\gG_j$ are defined by \eqref{6.5b} and
\eqref{6.5}, respectively. However, the Green's identity
\eqref{3.1AA} cannot extend to $\dom(A^*)$ in general, since
$\dom(\gG)$ is smaller than $\dom(A^*)$ generically. It might even
happen that $\gG_j$ are not bounded as mappings from $\dom(A^*)$
equipped with the graph norm into $\cH.$ Counterexamples such that
$\Pi=\bigoplus^\infty_{n=1} \Pi_n$ is not
a boundary triplet firstly appeared in
\cite{Koch79}).

In this section we show that it is always possible to modify the
boundary triplets $\gP_n$ in such a way that the new sequence $\wt
\gP_n= \{\cH_n,\wt\gG_{0},\wt\gG_{1}\}$ of boundary triplets for
$S^*_n$ such that $\wt \gP = \bigoplus^\infty_{n=1} \wt \gP_n$ defines a boundary
triplet for $A^*$ and the relations
\be\la{5.5A}
\wt S_{0n} := S^*_n\upharpoonright
\ker(\wt\gG_{0n}) = S^*_n\upharpoonright \ker(\gG_{0n}) =:
S_{0n}, \quad n\in \N,
\ee
are valid. Hence $\wt A_{0} := \bigoplus^\infty_{n=1} \wt S_{0n} = \bigoplus^\infty_{n=1}S_{0n} =: A_0$.
We note that the existence of a boundary triplet $\gP'= \{\cH, \gG'_{0},
\gG'_{1}\}$ for $A^*$ satisfying $\ker(\gG'_{0})= \dom(A_0) $ is
known (see \cite{DM91,GG91}). However, in applications we need a special boundary
triplet for $A^*$ which respects the direct sum structure and which
leads therefore to a block-diagonal form of
the corresponding Weyl function.
We start with a simple technical lemma.
\bl\la{VI.1}
Let $S$ be a densely defined closed symmetric operator with equal
deficiency indices, let $\gP = \{\cH,\gG_0,\gG_1\}$  be a boundary
triplet for $S^*$,  and let $M(\cdot)$ be the corresponding Weyl
function. Then there exists  a boundary triplet $\wt{\gP} =
\{\cH,\wt{\gG}_0,\wt{\gG}_1\}$ for $S^*$ such that
$\ker(\wt{\gG}_0) = \ker(\gG_0)$ and  the corresponding Weyl
function $\wt{M}(\cdot)$ satisfies  $\wt{M}(i) = i$.
\el
\begin{proof}
Let $M(i) = Q + iR^2$ where $Q := \re(M(i)),\ R :=
\sqrt{\im(M(i))}$. We set
\be\label{3.8}
\wt{\gG}_0 := R\gG_0
\quad \mbox{and} \quad
\wt{\gG}_1 := R^{-1}(\gG_1 - Q\gG_0).
\ee
A straightforward computation shows that $\wt{\gP} : =
\{\cH,\wt{\gG}_0,\wt{\gG}_1\}$ is a boundary triplet for $A^*$.
Clearly, $\ker(\wt{\gG}_0) = \ker(\gG_0)$. The Weyl function
$\wt{M}(\cdot)$ of $\wt{\gP}$ is given by $\wt{M}(\cdot) =
R^{-1}(M(\cdot) - Q)R^{-1}$ which yields $\wt{M}(i) = i$.
\end{proof}

If $S$ is a densely defined closed symmetric operator in $\gotH$,
then by the first v.~Neumann formula the direct decomposition
$\dom(S^*) = \dom(S) \stackrel{.}{+} \gotN_i \stackrel{.}{+}
\gotN_{-i}$  holds, where $\gotN_{\pm i} := \ker(S^*\mp i)$.
Equipping  $\dom(S^*)$ with the inner  product
\be\la{3.8b}
(f,g)_+ := (S^*f, S^*g) + (f,g), \quad f,g \in \dom(S^*),
\ee
one obtains a Hilbert space denoted by $\gotH_+$.
The first v.~Neumann formula leads to the following
orthogonal decomposition
\bed
\gotH_+ = \dom(S) \oplus \gotN_i \oplus \gotN_{-i}.
\eed
\bl\la{VI.2}
Let $S$, $\Pi$ and $M(\cdot)$ be as in Lemma \ref{VI.1}.
If $M(i) = i$, then $\gG: \gotH_+ \longrightarrow \cH \oplus \cH$,
$\gG:=(\Gamma_0,\Gamma_1)$  is a contraction. Moreover, $\gG$
isometrically maps $\gotN := \gotN_i\oplus \gotN_{-i}$ onto $\cH$.
\el
\begin{proof}
We show that
\be\la{6.4}
\|\gG(f + f_i + f_{-i})\|^2_{\cH \oplus \cH} = \|f_i + f_{-i}\|^2_+
\ee
where $f \stackrel{.}{+} f_i \stackrel{.}{+} f_{-i} \in \dom(S)
\stackrel{.}{+} \gotN_i \stackrel{.}{+} \gotN_{-i} = \dom(S^*)$.
Since $\dom(S) = \ker(\gG_0) \cap \ker(\gG_1)$ we
find
\bed
\|\gG(f + f_i + f_{-i})\|^2_{\cH \oplus \cH} =
\|\gG_0(f_i + f_{-i})\|^2_\cH + \|\gG_1(f_i + f_{-i})\|^2_\cH.
\eed
Clearly,
\be\la{6.4A} \|\gG_j(f_i + f_{-i})\|^2_\cH = \|\gG_jf_i\|^2  +
2\re((\gG_jf_i,\gG_jf_{-i})) +\  \|\gG_jf_{-i}\|^2, \quad j\in
\{0,1\}. \ee
Using $\gG_1f_i = M(i)\gG_0f_i = i\gG_0f_i$ and $\gG_1f_{-i} = M(-i)\gG_0f_{-i} = -i\gG_0f_{-i}$
we obtain
\be\la{6.4B}
\|\gG_1(f_i + f_{-i})\|^2_\cH =
(\gG_0f_i,\gG_0f_i) - 2\re((\gG_0f_i,\gG_0f_{-i})) + (\gG_0f_{-i},\gG_0f_{-i})
\ee
Taking a sum of \eqref{6.4A} and \eqref{6.4B} we get
     \be\la{6.4BC}
\|\gG_0(f_i + f_{-i})\|^2_\cH + \|\gG_1(f_i + f_{-i})\|^2_\cH =
2\|\gG_0f_i\|^2_\cH + 2\|\gG_0f_{-i}\|^2_\cH.
       \ee
Combining equalities $\gG_1f_{\pm i}= \pm i \gG_0f_{\pm i}$ with
Green's  identity \eqref{2.10A}  we obtain $\|\gG_0f_i\|_\cH =
\|f_i\|$ and $\|\gG_0f_{-i}\|_\cH = \|f_{-i}\|$. { Therefore
\eqref{6.4BC} takes the form}
\be\la{6.4C}
\|\gG_0(f_i + f_{-i})\|^2_\cH + \|\gG_1(f_i + f_{-i})\|^2_\cH =
2\|f_i\|^2 + 2\|f_{-i}\|^2.
\ee
A straightforward computation shows
$\|f_i + f_{-i}\|^2_+ = 2\|f_i\|^2 + 2\|f_{-i}\|^2$
which together with \eqref{6.4C} proves \eqref{6.4}.
Since $\|f_i + f_{-i}\|^2_+ \le \|f\|^2_+ + \|f_i + f_{-i}\|^2_+ =
\|f + f_i + f_{-i}\|^2_+$, we get from \eqref{6.4} that $\gG$ is a contraction.

Obviously, $\gG$ is an isometry from $\gotN$ into $\cH \oplus
\cH$. Since $\gP$ is a boundary triplet for $S^*$,
 $\ran(\gG) = \cH \oplus \cH.$ Hence $\gG$
is an isometry acting from $\gotN$ onto $\cH \oplus \cH$.
\end{proof}

Passing to  the direct sum \eqref{6.5a}, we  equip  $\dom(S^*_n)$ and
$\dom(A^*)$ with their graph's norms and obtain the Hilbert spaces
$\gotH_{+n}$ and $\gotH_+$, respectively. Clearly, the
corresponding inner products $(f,g)_{+n}$
and $(f,g)_+$ are defined by \eqref{3.8b} where $S^*$ is replaced by
$S^*_n$ and $A^*$, respectively. Obviously, $\gotH_+ =
\bigoplus^\infty_{n=1} \gotH_{+n}.$
\bt\la{VI.3}
Let $\{S_n\}^\infty_{n=1}$ be a sequence of densely defined
closed symmetric operators in $\gotH_n$
and let $S_{0n} = S^*_{0n}\in \Ext_{S_n}$. Further, let $A$ and $A_0$ be
given by \eqref{6.5a} and
\be\la{6.8}
A_0 := \bigoplus^\infty_{n=1}{ S_{0n}},
\ee
respectively. Then there exist boundary triplets $\gP_n :=
\{\cH_n,\gG_{0n},\gG_{1n}\}$ for $S^*_n$ such that ${
S_{0n}} = { S^*_n}\upharpoonright\ker(\gG_{0n})$, $n \in \N$,
and the direct sum ${ \gP = \bigoplus^\infty_{n=1} \gP_n}$
defined by \eqref{6.5} forms a boundary triplet for
$A^*$ satisfying $A_0 = A^*\upharpoonright\ker(\gG_0)$. Moreover,
the corresponding  Weyl function $M(\cdot)$ and the $\gga$-field
$\gga(\cdot)$ are { given by}
\be\la{6.7}
M(z) = \bigoplus^\infty_{n=1}M_n(z)
\quad \mbox{and} \quad
\gga(z) = \bigoplus^\infty_{n=1}\gga_n(z)
\ee
where $M_n(\cdot)$ and $\gga_n(\cdot)$ are the Weyl functions and
the $\gga$-field corresponding to $\gP_n, \ n\in \N$. In addition,
the condition { $M(i) = iI$}  { holds}.
\end{theorem}
\begin{proof}
For every ${ S_{0n}} = { S^*_{0n}}\in \Ext_{S_n}$
there exists a boundary triplet $\gP_n =
\{\cH_n,\gG_{0n},\gG_{1n}\}$ for $S_n^*$ such that $S_{0n}
:= S^*_n\upharpoonright\ker(\gG_{0n})$ (see \cite{DM91}). By Lemma
\ref{VI.1} we can assume without loss of generality that the
corresponding Weyl function $M_{n}(\cdot)$ satisfies $M_n(i) = i$.
By Lemma \ref{VI.2} the mapping
$\gG^n:=(\Gamma_{0n},\Gamma_{1n}): \gotH_{+n} \longrightarrow
\cH_n \oplus \cH_n$, is contractive { for each} $n\in \N$.
Hence $\|\gG_j\| = \sup_{n}\|\gG_{jn}\|\le 1,  j\in \{0,1\}$,
where $\gG_0$ and $\gG_1$ are defined by \eqref{6.5}. It follows
that the mappings $\gG_0$ and $\gG_1$ are well-defined on
$\dom(\gG) = \dom(A^*) = \bigoplus^\infty_{n=1}\dom({
S^*_n})$. Thus, { the Green's} identity \eqref{3.1AA} holds
for all $f,g \in \dom(A^*).$

Further, we set $\gotN_{\pm in} := \ker({ S^*_n} \mp i)$,\  $\gotN_n
:= \gotN_{in} \stackrel{.}{+} \gotN_{-in}$, $\gotN_{\pm i} :=
\ker(A^* \mp i)$ and $\gotN := \gotN_i \stackrel{.}{+}
\gotN_{-i}.$  By Lemma \ref{VI.2}  the restriction
$\gG^n\upharpoonright\gotN_n$ is an isometry from $\gotN_n,$
regarded as a subspace of $\gotH_{+n},$ onto $\cH_n \oplus \cH_n$.
Since $\gotN$ regarded as a subspace of $\gotH_+$ admits the
representation $\gotN = \bigoplus^\infty_{n=1}\gotN_n$, the
restriction $\gG\upharpoonright\gotN$, $\gG :=
\bigoplus^\infty_{n=1}\gG^n$, isometrically maps $\gotN$ onto $\cH
\oplus \cH$. Hence $\ran(\gG) = \cH \oplus \cH$. Equalities
\eqref{6.7} are follow  from  Definition \ref{Weylfunc}.
\end{proof}
\br
{\rm
Theorem \ref{VI.3} generalizes a result of
Kochubei \cite[Theorem 3]{Koch79} which states that for any sequence
of pairwise unitarily equivalent closed symmetric operators
$\{S_n\}_{n\in \N}$ there are boundary triplets $\Pi_n$ for $S^*_n$,
$n \in \N$ such that $\Pi = \bigoplus_{n\in \N}\Pi_n$ defines a boundary triplet for $A^* =
\bigoplus_{n\in\N}S^*_n$.
}
\er

Recall, that for any non-negative symmetric operator $A$
the set of its non-negative self-adjoint extensions
$\Ext_A(0,\infty)$ is non-empty (see \cite{AG81,Ka76}). The set
$\Ext_A(0,\infty)$ contains the Friedrichs (the biggest) extension
$A^F$ and the Krein (the smallest) extension $A^K$. These
extensions are uniquely determined by the following extremal
property in the class $\Ext_A(0,\infty):$
\bed
(A^F + x)^{-1}\le(\wt A+x)^{-1}\le(A^K + x)^{-1},\quad x>0,\quad
\wt A\in\Ext_A(0,\infty).
\eed
\begin{corollary}\label{cor5.5}
Let the assumptions of Theorem \ref{VI.3} be satisfied. Further, let $S_n\ge 0$,
$n\in{\N},$ and let $S^F_n$ and $S^K_n$ be the Friedrichs and
Krein extensions of $S_n$, respectively. Then
\begin{equation}\label{5.14}
A^F=\bigoplus^\infty_{n=1}S_n^F
\qquad \text{and}\qquad
A^K=\bigoplus^\infty_{n=1} S^K_n.
\end{equation}
\end{corollary}
\begin{proof}
Let us prove the second relation. The first one
is proved similarly.  By Theorem \ref{VI.3} there exists { a}
boundary triplet $\Pi_n=\{\cH_n,\Gamma_{0n},\Gamma_{1n}\}$ for
$S^*_n$ such that $S^K_n = S_{0n}$ and
$\Pi=\bigoplus^{\infty}_{n=1}\Pi_n$ is a boundary triplet for $A^*$.

Fix any $x_2\in\R_+$ and put $C_2:=\|M(-x_2)\|$. Then any
$h=\bigoplus^{\infty}_{n=1} h_n \in \cH$ can be decomposed by
$h=h^{(1)}\oplus h^{(2)}$ with $h^{(1)}\in\oplus^p_{n=1} \cH_n$ {
and} $h^{(2)}\in\oplus^{\infty}_{n=p+1}\cH_n$ {such that
$\|h^{(2)}\| < C^{-1/2}_2$}. Hence $|(M(-x_2)h^{(2)},h^{(2)})| <
1$. {  Due to the monotonicity of $M(\cdot)$ we get}
\bed
\bigg(M(-x)h^{(2)},h^{(2)}\bigg)>\bigg(M(-x_2)h^{(2)},h^{(2)}\bigg)>-1,\qquad
x\in (0, x_2).
\eed
Since $S_{0n} = S^K_n$, the Weyl function $M_n(\cdot)$ satisfies
\begin{equation}\label{5.16}
\lim_{{ x\downarrow 0}}\bigg(M_n(-x)g_n, g_n\bigg)=+\infty, \qquad
g_n\in\cH_n\setminus\{0\},
\end{equation}
{ cf. \cite[Proposition 4]{DM91}}. Because $M(\cdot)=
\bigoplus^\infty_{n=1} M_n(\cdot)$ is block-diagonal, cf.
\eqref{6.7}, we get from \eqref{5.16} that for any $N>0$ there
exists $x_1>0$ such that
\begin{equation}\label{5.17}
\bigg(M(-x)h^{(1)},h^{(1)}\bigg)=\sum^p_{n=1}\bigg(M_n(-x)h_n,h_n\bigg)>
N \quad \text{for}\quad x\in (0, x_1).
\end{equation}
Combining \eqref{5.16} with \eqref{5.17} and using the diagonal
form of $M(\cdot)$, we get
\bed (M(-x)h,h) = (M(-x)h^{(1)}, h^{(1)}) + (M(-x)h^{(2)},h^{(2)})
> N-1 \eed
for $0<x\le\min(x_1,x_2)$. Thus, $\lim_{{x\downarrow 0}}(M(-x)h,h)
= + \infty$ for $h\in \cH\setminus\{0\}.$  Applying
\cite[Proposition 4]{DM91} we {prove the second relation of}
\eqref{5.14}.
\end{proof}

\begin{remark}
{\em
Another proof can be obtained by using characterization of $A^F$
and $A^K$ by means of the respective quadratic forms.
}
\end{remark}

\subsection{Direct sums of symmetric operators with arbitrary deficiency indices}

We start with some simple spectral observations for direct sums of
symmetric operators where the symmetric operators may have arbitrary deficiency indices.
\begin{proposition}\la{VI.5b}
Let $\{S_n\}^\infty_{n=1}$ be a sequence of densely defined closed
symmetric operators in $\gotH_n$ and let $S_{0n} = S^*_{0n}\in \Ext_{S_n}$.
Further, let $A$ and $A_0$ be given by \eqref{6.5a} and
\eqref{6.8}, respectively. If $\wt A$ is a self-adjoint extension
of $A$ such that condition
\be\la{5.0}
(\wt A - i)^{-1} - (A_0 - i)^{-1} \in \gotS_\infty(\gotH)
\ee
is satisfied, then
\be\label{5.13}
\gs_{ac}(A_{0}) = \overline{\bigcup \gs_{ac}({ S_{0n}})} \subseteq
\gs(\wt A) \quad \mbox{and} \quad \gs_{ac}(\wt A) \subseteq
\overline{\bigcup \gs({ S_{0n}})} = \gs(A_{0}).
      \ee
\end{proposition}
\begin{proof}
By the Weyl theorem, condition \eqref{5.0} yields  $\gs_{\rm ess}(\wt
A) = \gs_{\rm ess}(A_0)$. Hence
\bed \overline{\bigcup \gs_{ac}({ S_{0n}})} = \gs_{ac}(A_0) \subseteq
\gs_{\rm ess}(A_0) = \gs_{\rm ess}(\wt A) \subseteq \gs(\wt A) \eed
and
\bed
\gs_{ac}(\wt A) \subseteq \gs_{\rm ess}(\wt A) = \gs_{\rm ess}(A_0)
\subseteq \gs(A_0) = \overline{\bigcup \gs({ S_{0n}})}
\eed
which completes the proof.
\end{proof}

Applying Theorem \ref{V.5} the results of Proposition \ref{VI.5b}
can be improved as follows.
\bt\la{VI.4}
Let $\{S_n\}^\infty_{n=1}$ be a sequence of densely defined closed
symmetric operators in $\gotH_n$ and let $S_{0n} = S^*_{0n}\in \Ext_{S_n}$.
Further, let $\gP_n = \{\cH_n,\gG_{0n},\gG_{1n}\}$ be a
boundary triplet for ${ S^*_n}$  such that $S_{0n} =
S^*_n\upharpoonright\ker(\gG_{0n})$, $n \in \N$, and let
$M_n(\cdot)$ be the corresponding Weyl function. Moreover, let
$\gotm^+_n(t)$, $n \in \N$, be the invariant maximal normal function
for $\Pi_n$. Further, let $A$ and $A_0$ be given by \eqref{6.5a} and
\eqref{6.8}, respectively.

If $\gd$ is a Lebesgue measurable subset of $\R$ such that
$\sup_{n\in\N}\gotm^+_n(t) < +\infty$ for a.e. $t \in \gd$, then
for any  self-adjoint extension $\wt{A}$ of $A$ satisfying the condition \eqref{5.0}, the
absolutely continuous parts $\wt{A}^{ac}E_{\wt A}(\gd)$ and $A^{ac}_0E_{A_0}(\gd)$ are
unitarily equivalent. In particular, if $\gd = \R$, then the parts
$\wt{A}^{ac}$ and $A^{ac}_0$ are unitarily equivalent
and \eqref{5.13} is replaced  by $\gs_{ac}(A_{0})= \gs_{ac}(\wt A)$.
\et
\begin{proof}
Let $\wt \gP_n = \{\cH_n, \wt\gG_{0n}, \wt\gG_{1n}\}$ be a
boundary triplet for $S^*_n, \ n\in \N,$ defined according to
\eqref{3.8}, that is ${\wt\Gamma}_{0n}:= R_n\Gamma_{0n}$ and
${\wt\Gamma}_{1n}:=R^{-1}_n\bigl(\Gamma_{1n}- \re
(M_n(i))\Gamma_{0n}\bigr)$, where $R_n:=\sqrt{\im M_n(i))}.$ The
corresponding Weyl function ${\wt M}_n(\cdot)$ is
\bed
{\wt M}_n(z)=R^{-1}_n\bigl(M_n(z)- \re M_n(i)\bigr)R^{-1}_n, \quad
n\in \N.
\eed
Since  ${\wt M}_n(i)=i,\  n\in \N,$ by Theorem \ref{VI.3},
${\wt\Pi}=\bigoplus^{\infty}_{n=1}{\wt\Pi}_n=:\{\cH,{\wt\Gamma}_0,{\wt\Gamma}_1\}$
is a boundary triplet for $A^*=\bigoplus_{n=1}^{\infty} S^*_n$
satisfying $A^*\upharpoonright\ker{\wt\Gamma}_0=A_0:=\bigoplus^{\infty}_{n=1}
S_{0n}$.  By the definition of $\gotm^+_n(\cdot)$ one has
$\gotm^+_n(t) = \wt m^+_n(t) := \sup_{y \in (0,1]}\|\wt M_n(t+iy)\|$ for $t \in \R$, $n \in \N$. Since
$A_0 =\bigoplus_{n=1}^{\infty}S_{0n}$ we get  that $\wt
m^+(t) = \sup_n m^+_n(t)$, where $\wt m^+(t) := \sup_{y \in
(0,1]}\|\wt M(t + iy)\|$, $t \in \R$. By assumption, the
maximal normal function $\wt m^+(t)$ is finite for a.e. $t \in \gd$. Hence we obtain from
Theorem \ref{V.5} that $\wt A^{ac}E_{\wt A}(\gd)$ and $A^{ac}_0E_{A_0}(\gd)$ are unitarily
equivalent.
\end{proof}

Let $T$ and $T'$ be densely defined closed symmetric operators in
$\gotH$ and let $T_0$ and $T'_0$ be self-adjoint extensions of
$T$ and $T'$, respectively. The pairs $\{T,T_0\}$
and $\{T',T'_0\}$ are called unitarily equivalent if there exists  a
unitary operator $U$ in $\gotH$ such that $T' = UTU^{-1}$ and
$T'_0 = UT_0U^{-1}$.
\bc\la{VI.5a}
Let the assumptions of Theorem \ref{VI.4} be satisfied. Moreover, let
the pairs $\{S_n,S_{0n}\}$, $n\in \N$, be unitarily
equivalent to the pair $\{S_1,S_{01}\}$. If the maximal
normal function $m^+_1(t) := \sup_{0 < y\le 1}\|M_1(t+iy)\|$ is finite for a.e. $t \in \gd$ and if the
condition \eqref{5.0} is satisfied, then the absolutely continuous
parts $\wt{A}^{ac}E_{\wt A}(\gd)$ and $A^{ac}_0E_{A_0}(\gd)$ are unitarily equivalent.
\ec
\begin{proof}
Since the symmetric operators ${ S_n}$ are unitarily
equivalent, we assume without loss of generality that $\cH_n =
\cH$ for each $n \in \N$. Let $U_n$ be a unitary operator such
that $A_1 = U_n{ S_n}U^{-1}_n$ and $A_{01}  = U_n{
S_{0n}}U^{-1}_n$. A straightforward computation shows that $\gP'_n
:= \{\cH,\gG'_{0n},\gG'_{1n}\}$, $\gG'_{0n} := \gG_{01}U_n$ and
$\gG'_{1n} := \gG_{1n}U_n$, defines a boundary triplet for ${
S^*_n}$. The Weyl function $M'_n(\cdot)$ corresponding to $\gP'_n$
is $M'_n(z) = M_1(z)$. Hence  $\gotm^+_n(\cdot)=
\gotm'^+_n(\cdot)$ and  $\gotm^+_1(t) = \gotm'^+_n(t)$ for $t \in
\R,$ where $\gotm^+_n(t)$ and $\gotm'^+_n(t)$ are the invariant
maximal normal functions corresponding  to the triplets $\gP_n$
and $\gP'_n$, respectively. Since $\gotm^+_1(t) = \gotm^+_n(t)$ for
$t\in \R$ and $n \in \N$ we complete the proof applying Theorem
\ref{VI.4}.
\end{proof}

\subsection{Direct sums of symmetric operators with finite deficiency}

Here we improve the previous results assuming that
$n_{\pm}({ S_n})<\infty.$ First, we show that extensions
$A_0=\bigoplus^\infty_{n=1} { S_{0n}}(\in \Ext_A)$ of the form
\eqref{6.8} possess a certain spectral minimality property. To this
end we start with the following lemma.
\bl\la{V.9}
Let $H$ be  a bounded non-negative self-adjoint operator in a
separable Hilbert  space $\gotH$ and let $L$ be a bounded operator
in $\gotH$. Then

\item[\rm \;\;(i)]  $\dim(\overline{\ran(H)}) =
\dim(\overline{\ran(\sqrt{H})})$;

\item[\rm \;\;(ii)]
If $L^*L  \le H$, then
$\dim(\overline{\ran(L)}) \le \dim(\overline{\ran(H)})$;

\item[\rm \;\;(iii)] If $P$ is an orthogonal projection,
then $\dim(\overline{\ran(PHP)}) \le
\dim(\overline{\ran(H)})$.
\el
\begin{proof}
The assertion (i) is obvious.

(ii)  If $L^*L \le H$, then there is a contraction $C$ such that
$L = C\sqrt{H}.$ Hence $\dim(\overline{\ran(L)}) =
\dim(\overline{\ran(C\sqrt{H})}) \le
\dim(\overline{\ran(\sqrt{H})}) = \dim(\overline{\ran(H)})$.

(iii) Clearly, $\dim(\overline{\ran(PHP)}) \le
\dim(\overline{\ran(\sqrt{H}P)}) \le
\dim(\overline{\ran(\sqrt{H})})$. Applying (i) we complete the proof.
\end{proof}

We are going to show that if the summands have only finite deficiency
indices, then the absolutely spectrum of extensions of the direct sum
can only increase comparing with the absolutely continuous spectrum of
those extensions which are direct sums of extensions.
\bt\la{VI.6}
Let $\{{ S_n}\}^\infty_{n=1}$ be a sequence of densely defined
closed symmetric operators in $\gotH_n$ and let
$S_{0n} = S^*_{0n}\in \Ext_{S_n}$. Further, let $A$
and $A_0$ be given by \eqref{6.5a} and \eqref{6.8}, respectively.

If the deficiency indices of $S_n$ are finite for each $n \in \N$, then $A_0$ is $ac$-minimal, in particular,
$\sigma_{ac}(A_0) \subseteq \sigma_{ac}(\wt A)$ for any
self-adjoint extension $\wt A$ of $A$.
\et
\begin{proof}
By Theorem \ref{VI.3} there is a sequence of boundary triplets
$\gP_n := \{\cH_n,\gG_{0n},\gG_{1n}\}$, $n \in \N$, for ${
S^*_{n}}$ such that ${S_{0n}}  = {
S^*_n}\upharpoonright\ker(\gG_{0n})$, $n \in \N$, and the direct
sum $\gP =  \{\cH,\gG_0,\gG_1\} = \bigoplus^\infty_{n=1}\gP_n$ of
the form \eqref{6.5a} is a boundary triplet for $A^*$ satisfying
$A_0 = A^*\upharpoonright\ker(\gG_0)$. By Proposition
\ref{prop2.1}, any $\wt A = {\wt A}^*\in \Ext_A$ admits a
representation  $\wt A = A_\gT$ with  $\gT = \gT^*\in \wt
\cC(\cH).$   By \cite[Corollary 4.2(i)]{MN2011},  we can
{ assume that } $\wt A$ and
$A_0,$ { are disjoint, that is }
{ $\gT= B =B^*\in \cC(\cH)$}. Consider the generalized Weyl
function $M_B(\cdot) := (B - M(\cdot))^{-1}$, where $M(\cdot) =
\bigoplus^\infty_{n=1}M_n(\cdot)$, cf. \eqref{6.7}. Clearly,
\bed
\imag(M_B(z)) = M_B(z)^*\imag(M(z))M_B(z), \quad z \in \C_+.
\eed
Denote by $P_N$, $N \in \N$,  the orthogonal projection from $\cH$
onto the subspace $\cH_N := \bigoplus^N_{n=1}\cH_n$. Setting
$M^{P_N}_B(z):= P_NM_B(z)\upharpoonright\cH_N$,  { and taking
into account the block-diagonal form of $M(\cdot)$ and the
inequality $\imag(M(z))>0$ we obtain}
\bea\la{5.18}
\lefteqn{
\imag(M^{P_N}_B(z)) = \imag(P_NM_B(z)P_N) }\\
& & = P_NM_B(z)^*\imag(M(z))M_B(z)P_N \ge
M^{P_N}_B(z)^*\imag(M^{P_N}(z))M^{P_N}_B(z), \nonumber \eea
where $M^{P_N}(z) := P_NM(z)\upharpoonright\cH_N =
\bigoplus^N_{n=1}M_n(z)$. Since $P^N$ is a finite dimensional
projection the limits $M^{P_N}_B(t) := \slim_{y\downarrow 0}M^{P_N}_B(t+iy)$
and $M^{P_N}(t) := \slim_{y\downarrow 0}M^{P_N}(t+iy)$ exists for a.e. $t \in
\R$. From \eqref{5.18} we get
\be\la{5.19}
\imag(M^{P_N}_B(t)) \ge
M^{P_N}_B(t)^*\imag(M^{P_N}(t))M^{P_N}_B(t) \quad\text{ for
a.e.}\quad t \in \R.
\ee
Since $M_B(\cdot)$ is { a} generalized Weyl function, it is a
strict $R_\cH$-function, that is, $\ker(\imag(M_B(z)))=\{0\},\
z\in \C_+$. Therefore, $M^{P_N}_B(\cdot)$ is also strict. Hence
$0\in \varrho(M^{P_N}_B(z))$, $z\in \C_+$, and $G_N(\cdot) := -
(M^{P_N}_B(\cdot))^{-1}$ { is strict}. Since both $G_N(\cdot)$ and
$M^{P_N}_B(\cdot)$ are { matrix-valued} $R$-functions, the limits
$M^{P_N}_B(t+i0):= \lim_{y\downarrow 0} M^{P_N}_B(t+iy)$ and $G_N(t+i0):=
\lim_{y\downarrow 0} G_N(t+iy)$ exist for a.e. $t\in \R.$ Therefore,
passing to the limit in the identity $M^{P_N}_B(t+iy)G_N(t+iy) =
-I$ as $y\to 0,$  we get $M^{P_N}_B(t+i0)G_N(t+i0) = -I$ for a.e.
$t\in \R.$ Hence $M^{P_N}_B(t) := M^{P_N}_B(t+i0)$ is invertible
for a.e. $t \in \R$.

Further, combining  \eqref{5.19} with Lemma \ref{V.9}(ii) we get
     \bed
\dim\left(\overline{\ran\left(\sqrt{\imag
M^{P_N}(t)}M^{P_N}_B(t)\right)}\right) \le d_{M^{P_N}_B}(t) \quad
\text{for a.e.} \quad t \in \R.
        \eed
Since $M^{P_N}_B(t)$ is invertible for a.e. $t \in \R,$ we find
      \be\label{5.21}
 d_{M^{P_N}}(t) := \dim\left(\overline{\ran \left( \sqrt
{\imag M^{P_N}(t)}\right)}\right) \le d_{M^{P_N}_B}(t)\quad
\text{for a.e.} \quad t \in \R.
       \ee
Let ${ D}_N = P_N \oplus D_0$ where { ${D}_0 \in
\gotS_2(\cH^\perp_N)$ and satisfy} $\ker({ D}_0) = \ker({D}^*_0) =
\{0\}$. Then $\ker({ D}_N) = \ker({D}^*_N) = \{0\}$ and $P_N =
P_N{D}_N = { D}_NP_N$. By Lemma \ref{V.9}(iii),
$d_{M^{P_N}}(t) \le d_{M^{{ D}_N}_B}(t)$ for a.e. $t \in \R$.
Further, for any ${D}\in \gotS_2(\cH)$ and satisfying $\ker({D}) =
\ker({D}^*) = \{0\},$ $d_{M_B^{D}}(t) = d_{M^{{D}_N}_B}(t)$ for
a.e. $t \in \R.$ { Combining this equality with \eqref{5.21}}
we get  $d_{M^{P_N}}(t) \le d_{M_B^{ D}}(t)$ for a.e. $t \in \R$
and $N \in \N$. Since
\be\la{6.13A}
d_{M^{P_N}}(t) = \sum^N_{n=1}d_{M_n(t)} \quad \mbox{and} \quad
d_{M^D}(t) = \sum^\infty_{n=1}d_{M_n}(t)
\ee
for a.e. $t \in \R,$ we finally prove that $d_{M^{ D}}(t) \le
d_{M_B^{D}}(t)$ for a.e. $t \in \R$. One completes the
proof by applying Theorem \ref{IV.10}(i).
\end{proof}

Taking into account Proposition \ref{III.8}  and Corollary \ref{IV.9}
the proof of Theorem \ref{VI.6} shows us that in fact the spectral multiplicity
function $N_{\wt A^{ac}}(t)$ can only be increase with respect to
$N_{A^{ac}_0}(t)$, that is, one always has $N_{\wt A^{ac}}(t) \ge
N_{A^{ac}_0}(t)$ for a.e. $t \in \R$ and any self-adjoint extension
$\wt A$ of $A$.

\bc\la{VI.7A}
Let the assumptions of Theorem \ref{VI.6} be satisfied. If $S_n \ge 0$, $n \in \N$ and if the deficiency
indices of $S_n$ are finite for each $n \in \N$, then the
Friedrichs and the Krein extensions $A^F$ and $A^K$ of
$A$ are $ac$-minimal. In particular, $(A^F)^{ac}$ and $(A^K)^{ac}$
are unitarily equivalent.
\ec
\begin{proof}
Combining  Theorem \ref{VI.6}  and Corollary \ref{cor5.5} one
immediately proves the assertions.
\end{proof}
\bc\la{VI.7}
Let the assumptions of Theorem \ref{VI.4} be satisfied. Further, let the deficiency indices of
$S_n$ be finite for each $n \in \N$.

\item[\;\;\rm (i)]
If
\be\la{6.13}
\gd_\infty := \{t \in \R: \sum_{n\in\N}d_{M_n}(t) = \infty\},
\ee
then for any self-adjoint extension $\wt A$ of $A$ the parts $\wt A^{ac}E_{\wt A}(\gd_\infty)$ and
$A^{ac}_0E_{A_0}(\gd_\infty)$ are unitarily equivalent.

\item[\;\;\rm (ii)]
If $\gd$ is a Lebesgue measurable subset of $\R$ such that $\sup_n\gotm^+_n(t) < \infty$ for a.e. $t \in \gd$, then
for any self-adjoint extension $\wt A$ of $A$ the parts $\wt A^{ac}E_{\wt A}(\gd_\infty \cup \gd)$ and
$A^{ac}_0E_{A_0}(\gd_\infty\cup\gd)$ are unitarily equivalent.
\ec
\begin{proof}
(i)
By \eqref{6.13A} and \eqref{6.13} we find $d_{M^D}(t) = + \infty$ for
a.e. $t \in \gd_\infty$.  Since by Theorem \ref{VI.6} the spectral multiplicity function
can only be increase for self-adjoint extensions $\wt A$ one gets that
$N_{\wt A^{ac}}(t)  = N_{A^{ac}_0}(t)$ for a.e. $t \in \gd$ which immediately yields
the unitary equivalence of the parts $\wt A^{ac}E_{\wt A}(\gd_\infty)$ and
$A^{ac}_0E_{A_0}(\gd_\infty)$.

(ii)
By Theorem \ref{VI.4} the parts $\wt A^{ac}E_{\wt A}(\gd)$ and
$A^{ac}_0E_{A_0}(\gd)$ are unitarily equivalent. Using (i) we immediately obtain
the unitary equivalence of the parts
$\wt A^{ac}E_{\wt A}(\gd_\infty \cup \gd)$ and $A^{ac}_0E_{A_0}(\gd_\infty\cup\gd)$.
\end{proof}
\bc\la{VI.11}
Let the assumptions of Theorem \ref{VI.6} be satisfied.
If the deficiency indices of $S_n$ are finite for each $n \in \N$,
then $\overline{\bigcup_{n\in\N}\gs_{ac}(S_{0n})} \subseteq
\gs_{ac}(\wt A)$ for any self-adjoint extension $\wt A$ of $A$.
If in addition condition \eqref{5.0} is valid
and the extensions $S_{0n}$ are purely
absolutely continuous for each $n \in \N$, then
\be\la{6.17}
\gs_{ac}(\wt A) = \overline{\bigcup_{n\in\N}\gs_{ac}({ S_{0n}})}.
\ee
\ec
\begin{proof}
The first statement immediately follows from  Theorem \ref{VI.6}.
Relation \eqref{6.17} is implied by  { Proposition} \ref{VI.5b}.
\end{proof}
\bc\la{VI.12}
Let the assumptions of Theorem \ref{VI.6} be satisfied. Further, let
the pairs $\{S_n,S_{0n}\}$, $n \in \N$, be
unitarily equivalent to $\{S_1,S_{01}\}$. If the deficiency indices of $S_n$ are finite for each $n \in \N$,
holds, then for any self-adjoint extension $\wt A$ of $A$ satisfying
condition \eqref{5.0} the $ac$-parts $\wt A^{ac}$ and $A^{ac}_0$ are unitarily equivalent.
\ec
\begin{proof}
The proof follows immediately from Corollary \ref{VI.5a}.
\end{proof}
\br
{\rm
(i) For { the} special case { $n_\pm(S_n) = 1$, $n \in
\N$}, Theorem \ref{VI.6} complements \cite[Corollary 5.4]{ABMN05}
where the inclusion { $\sigma_{ac}(A_0)\subseteq
\sigma_{ac}(\wt A)$} was proved. Moreover, Corollary  \ref{VI.12}
might be { regarded} as a substantial generalization of
\cite[Theorem 5.6(i)]{ABMN05} to the case $n_{\pm}({ S_n})>1$.
However,  in the case $n_{\pm}({ S_n})=1,$ Corollary
\ref{VI.12} is implied by \cite[Theorem 5.6(i)]{ABMN05} where the
unitary equivalence of ${\wt A}^{ac} = {\wt A}_B^{ac}$ and
$A_0^{ac}$ was proved  under { the weaker} assumption that
$B$ is purely singular. Indeed, by Proposition \ref{prop2.9}
condition \eqref{5.0} with $\wt A = A_B$  is equivalent to
the discreteness of $B$.

(ii) The inequality   $N_{A^{ac}_0}(t) \le
N_{\wt{A}^{ac}}(t)$ in Theorem \ref{VI.6} might be strict even
for $t\in\sigma_{ac}(A_0)$. Indeed, assume that $(\alpha,\beta)$
is a gap for all { except for the operators}
$S_1,\ldots,S_N$. Set $S_1:=\bigoplus^N_{n=1} { S_n}$ and
$S_2:=\bigoplus^{\infty}_{ n=N+1}{ S_n}$.  Then
$n_{\pm}(S_2)=\infty$ and $(\alpha,\beta)$ is a gap for $S_2.$ By
\cite{Bra04} there exists ${\wt S}_2={\wt S}^*_2\in \Ext_{S_2}$
having $ac$-spectrum within $(\alpha,\beta)$ of arbitrary
multiplicity. Moreover, even for operators $A =
\bigoplus^\infty_{n=1} S_n$  satisfying assumptions of
Corollary \ref{VI.12} with $n_{\pm}({ S_n})=1$ the inclusion
{ $\gs_{ac}(A_0)\subseteq \gs_{ac}(\wt A)$} might be strict
whenever condition \eqref{5.0} is violated, { cf. \cite{Bra04}
or \cite[Theorem 4.4]{ABMN05} which  guarantees the
appearance of prescribed spectrum either within one gap or within
several gaps of $A_0$}.
}
\er

\section{Sturm-Liouville operators with bounded operator potentials}

Let $\cH$ be an infinite dimensional separable Hilbert space. As usual, $L^2({\mathbb
R}_+,\cH)$ stands for the Hilbert
space of (weakly) measurable vector-functions $f(\cdot):{\mathbb
R}_+\to\cH$ satisfying $\int_{{\mathbb
R}_+}\|f(t)\|^2_{\cH}dt<\infty$. Denote also by $W^{2,2}({\mathbb
R}_+,\cH)$ the Sobolev space of
vector-functions taking  values in $\cH$.

Let $T = T^* \ge 0$ be a bounded operator in $\cH$.  Denote by
$A:=A_{\min}$ the minimal operator generated by $\cA$,
cf. \eqref{8.10}, in $\gH:=L^2({\R}_+,\cH)$.
It is known (see \cite{GG91,Rof-Bek69}) that the minimal operator $A$ is
given by
\begin{eqnarray}\la{8.1}
(Af)(x)   =  -\frac{d^2}{dx^2}f(x) + Tf(x), \quad f \in \dom(A) =
W^{2,2}_0({\mathbb R}_+,\cH),
\end{eqnarray}
where $W^{2,2}_0(\R_+,\cH) := \{f \in W^{2,2}(\R_+,\cH):\ f(0) =
f'(0) =0 \}$.

The operator $A$ is closed, symmetric and non-negative. It can be
proved similarly to \cite[Example 5.3]{BMN02} that $A$ is simple.
The adjoint operator $A^*$ is given by \cite[Theorem 3.4.1]{GG91}
\be\la{8.1B}
(A^*f)(x) = -\frac{d^2}{dx^2}f(x) + Tf(x), \quad f \in \dom(A^*) =
W^{2,2}(\R_+,\cH).
\ee
The Dirichlet realization $A^D$ is defined by $A^Df := \cA f$, $f \in
\dom(A^D) := \{g \in W^{2,2}(\R_+,\cH): g(0) = 0\}$.
Similarly, the Neumann realization $A^N$ is defined by $A^Nf := \cA f$,
$f \in \dom(A^N) := \{g \in W^{2,2}(\R_+,\cH: g'(0) = 0\}$.
Since $\dom(A) \subseteq \dom(A^D),\dom(A^N) \subseteq \dom(A^*)$
one gets that $A^D$ and $A^N$ are proper
extensions of $A$. One easily verifies that $A^D$ and $A^N$  are
symmetric extensions.

By \cite[Theorem 1.3.1]{LioMag72} the trace operators
$\Gamma_0, \ \Gamma_1: \dom(A^*)\to \cH$,
\be\la{8.2}
\gG_0 f = f(0) \quad \mbox{and} \quad \gG_1 f = f'(0),
\quad f\in \dom(A^*),
\ee
are well defined. Moreover, the deficiency subspace
$\mathfrak N_z(A)$ is
     \begin{equation}\label{8.3A}
\mathfrak N_z(A) = \{e^{ix\sqrt{z-T}}h: \ h \in \cH\}, \qquad z\in
\C_{\pm},
         \end{equation}
with the cut along $\R_+$.
\begin{lemma}\label{lem6.1}
A triplet $\gP = \{\cH,\gG_0,\gG_1\}$, where $\gG_0$ and $\gG_1$
are defined by  \eqref{8.2}, forms
 a boundary triplet for $A^*$.
The corresponding Weyl function $M(\cdot)$ is
\be\la{8.2B} M(z) = i\sqrt{z-T} = i\int\sqrt{z - \gl}\;dE_T(\gl),  \quad z \in \C_+.
 \ee
\end{lemma}
\begin{proof}
One obtains the Green formula integrating by parts. The
surjectivity of the mapping $\gG:=(\Gamma_0,\Gamma_1):
\dom(A^*) \rightarrow \cH \oplus \cH$ follows from
\eqref{8.2} and \cite[Theorem 1.3.2]{LioMag72}. Formula
\eqref{8.2B} is implied by \eqref{8.3A}.
\end{proof}
\bl\la{VII.1}
Let $T$ be a bounded non-negative self-adjoint operator in $\cH$
and let $A$ and $\Pi=\{\cH,\Gamma_0,\Gamma_1\}$ be defined by
\eqref{8.1} and \eqref{8.2}, respectively. Then

\item[\rm\;\;(i)] the invariant maximal normal function $\gotm^+(t)$
of the Weyl function $M(\cdot)$ is finite for all $t \in \R$ and
satisfies
\be\la{7.3a}
\gotm^+(t) \le (1 + \sqrt{2})(1 + t^2)^{1/4}, \quad t \in \R.
\ee

\item[\rm\;\;(ii)] The limit $M(t+i0) := \slim_{y\downarrow 0}M(t+iy)$
exists, is bounded and equals
\be\la{7.3}
M(t +i0) = i\int_\R \sqrt{t-\gl} dE_T(\gl) \quad \text{for any}\ \
t \in \R.
\ee
\item[\rm\;\;(iii)]  $d_M(t) = \dim(\ran(E_T([0,t))))$ for any $t
\in \R$.
\el
\begin{proof} (i)
It follows from \eqref{8.2B} and definition \eqref{4.12A} that
\bed
\gotm^+(t) \le \sup_{y\in(0,1]}\sup_{\gl \ge 0}
\left|\frac{\sqrt{t + iy -\gl} -
\real(\sqrt{i-\gl})}{\imag(\sqrt{i - \gl})}\right|.
\eed
Clearly, $\sqrt{i-\gl} = (1 + \gl^2)^{1/4}e^{i(\pi -\varphi)/2}$
where $ \varphi := \arccos\left(\tfrac{\gl}{\sqrt{1 +
\gl^2}}\right).$ Hence
\bed
\left|\frac{\re(\sqrt{i-\gl})}{\im(\sqrt{i-\gl})}\right| =
\tan\left(\frac{\varphi}{2}\right) = \frac{1}{\gl + \sqrt{1 + \gl^2}} \le 1,
\quad \gl \ge 0.
\eed
Furthermore, we have
\bed
\left|\frac{\sqrt{t + iy -\gl}}{\im(\sqrt{i - \gl})}\right|
\le \sqrt{2}\sqrt{\frac{\sqrt{(\gl - t)^2 + y^2}}{\gl + \sqrt{1 +
\gl^2}}} \le \sqrt{2}(1 + t^2)^{1/4}
\eed
for $\gl \ge 0$, $t \in \R$ and $y \in (0,1]$ which yields
\eqref{7.3a}.

(ii) From \eqref{8.2B} we find $M(t):= M(t +i0) := \slim_{y\downarrow 0}
i\sqrt{t+iy - T} = i\sqrt{t-T},$ for  any $t \in \R,$ which proves
\eqref{7.3}. Clearly, $M(t)\in [\cH]$ since $T\in [\cH].$

(iii) It follows that $\im(M(t)) = \sqrt{t-T}E_T([0,t))$, which
yields $d_M(t) = \dim(\ran(\im(M(t))))= \dim(\ran(E_T([0,t))))$.
\end{proof}

With $A=A_{\min}$ one associates a
closable quadratic form ${\mathfrak t}'_F[f]:= (Af,f)$,
$\dom({\mathfrak t}')=\dom(A)$. Its closure ${\mathfrak t}_F$ is
given by
\be\la{8.1A} { \mathfrak t_F[f]} := \int_{\R_+}
\left\{\|f'(x)\|^2_\cH + \|\sqrt{T}f(x)\|^2_\cH\right\}dx, \ee
$f \in \dom({\mathfrak t}_F) = W^{1,2}_0(\R_+,\cH)$,
where  $W^{1,2}_0(\R_+,\cH) := \{f \in W^{1,2}(\R_+,\cH): f(0)=0\}.$
By definition,  the Friedrichs extension $A^F$ of $A$ is a
self-adjoint operator associated with ${\mathfrak t}_F$. Clearly,
$A^F = A^*\upharpoonright (\dom(A^*)\cap\dom({\mathfrak t}_F)).$
\bt\label{prop7.2A}
Let $T\ge 0$, $T = T^*\in [\cH]$, and $t_0:=\inf\gs(T)$.
Let $A$ be defined by \eqref{8.1} and $\gP = \{\cH,\gG_0,\gG_1\}$
the boundary triplet for $A^*$ defined by \eqref{8.2}.
Then the following holds:

\item[\rm \;\;(i)]
The Dirichlet realization $A^D$ coincides with
$A_0 : = A^*\upharpoonright\ker(\gG_0)$ which is identical
with the Friedrichs extension $A^F$.
Moreover, $A^D$ is absolutely continuous and its
spectrum is given by $\gs(A^D)=\gs_{ac}(A^D)= [t_0,\infty)$.

\item[\rm\;\;(ii)]
The Neumann realization $A^N$
coincides with $A_1 := A^*\upharpoonright \ker(\gG_1)$.  $A^N$ is absolutely continuous
$(A^N)^{ac} = A^N$ and $\gs(A^N) = \gs_{ac}(A^N) =
[t_0,\infty)$.

\item[\rm \;\;(iii)] The Krein realization (or extension) $A^K$ is given by
\begin{equation}\label{7.4C}
\dom(A^K)=\{f\in W^{2,2}({\mathbb R}_+,\cH):\ f'(0) +
\sqrt{T}f(0)=0\}.
\end{equation}
Moreover, $\ker(A^K) = \gH_0 := \overline{\gH'_0}$,  $\gH'_0:=
\{e^{-x\sqrt{T}}h:\ h\in \ran(T^{1/4})\}$ and the restriction
$A^K\upharpoonright\dom(A^K) \cap \gH^{\perp}_0$ is absolutely
continuous, that is, $\gH^{\perp}_0=\gH^{ac}(A^{K})$ and $A^K=
{0}_{\gH_0}\oplus (A^K)^{ac}$. In particular, $\gs(A^K) = \{0\} \cup
\gs_{ac}(A^K)$ and $\gs_{ac}(A^K) = [t_0,\infty)$.

\item[\rm \;\;(iv)]  The realizations $A^D$, $A^N$ and $(A^K)^{ac}$ are
unitarily  equivalent.
\et
\begin{proof}
(i)  It follows from  \eqref{8.1B} and \eqref{8.2} that $\dom(A^D) = \dom(A_0)$
which yields  $A^D = A_0$.
Since $\dom(A_0) \subseteq W^{1,2}_0(\R_+,\cH)=
\dom({\mathfrak t}_F)$  we have $A^F = A_0$ (see \cite[Section
8]{AG81} and \cite[Theorem 6.2.11]{Ka76}).
It follows from \eqref{7.3} and \cite[Theorem 4.3]{BMN02} that
$\gs_p(A_0) = \gs_{sc}(A_0) = \emptyset$. Hence $A_0$ is
absolutely continuous. Taking into account Lemma \ref{VII.1}(iii)
and Proposition \ref{III.8} we get $\gs(A_0) = \gs_{ac}(A_0) =
\cl_{ac}(\supp(d_M)) = [t_0,\infty)$ which proves (i).

(ii) Obviously we have $\dom(A^N) = \dom(A_1) := \ker(\gG_1)$
which proves $A^N = A_1$. It follows from  Lemma \ref{lem6.1} and
\eqref{2.5} that the Weyl function corresponding to $A_1$ is given by
   \be\label{7.14B}
M_0(z) := (0 -M(z))^{-1} = i(z - T)^{-1/2} =
i\int \frac{1}{\sqrt{z -\gl}}dE_T(\gl), \quad z \in \C_+.
  \ee
Since  $M_0(\cdot)$ is regular within $(-\infty, t_0),$  we have
$(-\infty, t_0)\subset \varrho(A_1)$. Further, let $\tau
> t_0$. We set $\cH_\tau := E_T([t_0,\tau))\cH$ and note that for any $h \in
\cH_\tau$ and $t > \tau$
\begin{equation}\label{7.14}
\bigl(M_0(t+i0)h,h\bigr) = i\bigl((t-T)^{-1/2}h,h\bigr)=
i\int_{t_0}^{\tau} \frac{1}{\sqrt{t -\gl}}d(E_T(\gl)h,h).
\end{equation}
Hence  for any $h\in \cH_\tau\setminus\{0\}$ and   $t > \tau$
\bed
0 <
(t-t_0)^{-1/2}\|h\|^2 \le \imag(M_0(t+i0)h,h) = \int_{t_0}^{\tau}
(t -\gl)^{-1/2}d(E_T(\gl)h,h) <\infty.
\eed
By \cite[Proposition 4.2]{BMN02},  $\gs_{ac}(A_1) \supseteq
[\tau,\infty)$ for any $\tau > t_0,$ which yields $\gs_{ac}(A_1) =
[t_0,\infty)$. It remains to show that $A_1$ is purely
absolutely continuous. Since $M_0(t+i0)\not\in [\cH]$ we cannot
apply \cite[Theorem 4.3]{BMN02}. Fortunately, to we can use \cite[Corollary
4.7]{BMN02}. For any   $t \in \R,$ \ $y > 0,$ and $h \in \cH$ we
set
\bed V_h(t+iy) := \im(M_0(t+iy)h,h) =
\int\im\left(\frac{1}{\sqrt{\gl-t-iy}}\right)d(E_T(\gl)h,h). \eed
Obviously, one has
\bed V_h(t+iy) \le \int\frac{1}{((\gl-t)^2 +
y^2)^{1/4}}d(E_T(\gl)h,h), \quad t \in \R, \quad y > 0, \quad h
\in \cH. \eed
Hence
\bed V_h(t+iy)^p \le \|h\|^{2(p-1)} \int\frac{1}{((\gl-t)^2 +
y^2)^{p/4}}d(E_T(\gl)h,h), \quad p \in (1,\infty). \eed
We show that  for $p \in (1,2)$ and $-\infty < a < b < \infty$
\bed C_p(h; a,b) := \sup_{y\in (0,1]}\int^b_a V_h(t+iy)^p\;dt <
\infty. \eed
Clearly,
\bead
\int^b_a V_h(t+iy)^p dt \le \|h\|^{2(p-1)}\;\int^{\|T\|}_0
d(E(\gl)h,h) \int^b_a \frac{1}{((\gl-t)^2 + y^2)^{p/4}}dt \nonumber  \\
= \|h\|^{2(p-1)}\;\int^{\|T\|}_0 d(E(\gl)h,h) \int^{b-\gl}_{a-\gl}
\frac{1}{(t^2 + y^2)^{p/4}}dt. \quad
\eead
Note, that for $p \in (1,2)$ and  $-\infty < a < b < \infty$
  \begin{equation*}
\int^{b-\gl}_{a-\gl} \frac{1}{(t^2 + y^2)^{p/4}}dt \le
\int^{b}_{a-\|T\|} \frac{1}{t^{p/2}}dt =: \varkappa_p(b,a-\|T\|) <
\infty,
     \end{equation*}
Hence  $C_p(h;a,b) \le \varkappa_p(b,a-\|T\|)\|h\|^{2p} < \infty$
for $p \in (1,2)$, $-\infty < a < b < \infty$ and $h \in \cH$. By
\cite[Corollary 4.7]{BMN02},  $A_1$ is purely absolutely
continuous on any bounded interval $(a,b)$. Hence $A_1$ is purely
absolutely continuous.

(iii)\ By \cite[Proposition 5]{DM91}  $A^K$ is defined by $A^K=
A^*\upharpoonright \ker(\gG_1 - M(0)\gG_0).$ It follows from \eqref{8.2B}
that $M(0)=-\sqrt{T}$. Therefore, $A^K$ is defined by \eqref{7.4C}.

It follows from the extremal property of the Krein extension that
$\ker(A^K) =\ker(A^*)$. Clearly,   $f_h(x) := \exp(-x\sqrt{T})h\in
L^2(\R_+,\cH),$  $h\in \ran(T^{1/4}),$ since
\bead
\lefteqn{
\int^{\infty}_0\!\|\exp(-x\sqrt{T})h\|^2_{\cH}dx}\\
& &
=
\int^{\|T\|}_0\!d\rho_h(t)\int^{\infty}_0
\!e^{-2x\sqrt{t}}dx=\int^{\|T\|}_0 \!\frac{1}{2\sqrt
t}d\rho_h(t)<\infty,
\eead
where $\rho_h(t):=\bigl(E_T(t)h,h\bigr)$. Thus,
$\gH'_0\subset \ker(A^*).$ It is easily seen that  $\gH'_0$ is
dense in $\gH_0.$ To investigate the rest of the spectrum of
$A^K$ consider the Weyl function $M_K(\cdot)$ corresponding to
$A^K$. It follows from \eqref{8.2B} and \eqref{2.5} that
\bead
\lefteqn{
M_K(z)=M_{-\sqrt{T}}(z)=-\bigl(\sqrt{T}+M(z)\bigr)^{-1}}\\
& &
= -(\sqrt{T}+i\sqrt{z-T})^{-1}
=\frac{1}{z}(i\sqrt{z-T}-\sqrt{T}) = -\frac{2\sqrt{T}}{z}+\Phi(z).
\nonumber
\eead
where $\Phi(z):=\frac{1}{z}[i\sqrt{z-T}+\sqrt{T}]$.
For $t>0$ we get
\begin{equation}\label{7.13}
\im M_K(t+i0) =  \im\Phi(t+i0)=t^{-1}\sqrt{t-T}E_T([0,t)).
\end{equation}
Hence, by \cite[Theorem4.3]{BMN02},  $\sigma_p(A^K)\cap(0,\infty)=
\sigma_{sc}(A^K)\cap(0,\infty)=\emptyset$. It follows from
\eqref{7.13} that $\im{(M_K(t+i0))}>0$ for $t>t_0$. By Corollary
\ref{IV.9} we find $\sigma_{ac}(A^K)= [t_0,\infty)$.

(iv) It follows from \eqref{7.3} and \eqref{7.13} that
$d_M(t)=d_{M_K}(t)= \dim(\ran(E_T([0,t))))$ for  $t>t_0$.
Combining this equality with $\sigma_{ac}(A^K)= \sigma_{ac}(A^F)=
[t_0,\infty),$ we conclude from Theorem \ref{IV.10}(ii)
that $A^F$ and $(A^K)^{ac}$ are unitarily equivalent.

Passing to $A_1,$ we assume that $1 \le \dim(\ran(E_T([0,s))))=p_1
< \infty$ for some $s > 0$. Let $\gl_k$, $k \in \{1,\ldots,p\}, \
p\le p_1$, be the set of distinct eigenvalues within $[0,s)$.
Since $M_0(t+iy)E_T([0,t))$ is the $p\times p$ matrix-function,
the limit $M_0(t+i0)E_T([0,t))$ exists for $t \in [0,s) \setminus
\bigcup^p_{k=1}\{\gl_k\}.$   It follows from \eqref{7.14} that
\bed
\im(M_0(t)) = |T-t|^{-1/2}E_T([0,t)),\qquad t \in [0,s)
\setminus \bigcup^p_{k=1}\{\gl_k\}.
\eed
This yields
\bed
d_{M_0(t))} := \dim(\ran(\im(M_0(t)))) =
\dim(\ran(E_T([0,t)))) = d_M(t)
\eed
for a.e $t \in [0,s) \setminus \bigcup^p_{k=1}\{\gl_k\}$, that
is, for a.e. $t \in [0,s)$.

If $\dim(E_T([t_0,s))) = \infty$, then there exists  a point
$s_0\in (0, s),$ such that $\dim(E_T([0,s_0])) = \infty$ and
$\dim(E_T([0,s))) < \infty$ for $s \in [0,s_0)$. For any  $t \in
(s_0,s)$ choose  $\tau \in (s_0,t)$ and note that
$\dim(\ran(E_T([0,\tau)))) = \infty$. We set $\cH_\tau :=
E_T([0,\tau))\cH$ and $\cH_\infty := E_T([\tau,\infty))\cH$ as
well as $T_\tau := TE_T([0,\tau))$ and $T_\infty :=
TE_T([\tau,\infty))$. Further, we choose Hilbert-Schmidt
operators $D_\tau$ and ${D}_\infty$ in $\cH_\tau$ and
$\cH_\infty$, respectively, such that $\ker(D_\tau) =
\ker(D^*_\tau) = \ker(D_\infty) = \ker(
D^*_\infty) = \{0\}$. According to the decomposition
$\cH=\cH_\tau \oplus\cH_\infty$ we have $M_0 = M_\tau
\oplus M_\infty,$ $D = D_\tau \oplus D_\infty$ and
$d_{M_0^{ D}}(t) = d_{M_\tau^{{ D}_\tau}}(t) + d_{M_\infty^{
D_\infty}}(t)$ for a.e. $t \in [0,\infty).$
Hence $d_{M_0^{D}}(t) \ge d_{M_\tau^{{ D}_\tau}}(t)$ for a.e. $t
\in [0,\infty)$. Clearly, $M_\tau (t+iy) = i(t + iy -
T_\tau)^{-1/2}.$ If $t > \tau$, then $t\in \varrho(T_\tau)$ and
$M(t) := \slim_{y\downarrow 0} M(t+i0)$ exists and
\bed
M_\tau(t) := \slim_{y\to 0}M_\tau(t+iy) = i(t -
T_\tau)^{-1/2}E_T([0,\tau)).
\eed
Hence $d_{M_\tau^{D_\tau}}(t) = \dim(\ran(E_T([0,\tau)))) =
\infty$ for $t > s_0$. Hence $d_{M_0^D}(t) = d_M(t) = \infty$ for a.e.
$t > s_0$ which yields $d_{M_0^D}(t) = d_M(t)$ for a.e. $t \in
[0,\infty)$. Using Theorem \ref{IV.10}(ii) we obtain that $A^{ac}_0$
and $A^{ac}_1$  are unitarily equivalent which shows $A_0$ and $A_1$
are unitarily equivalent.
\end{proof}
\br\la{IV.4}
{\em
The statements on $A^D$, $A^N$ and $A^K$ are proved self-consistently
in the framework of boundary triplets.
However, the unitary equivalence of $A^D$ and $A^N$ can be proved much simpler. In fact,
the Dirichlet and Neumann realizations $l_D$ and $l_N$ of the differential
expression $l := -\frac{d^2}{dt^2}$ in $L^2(\R_+)$ are unitary equivalent.
If $U:\ L^2(\R_+)\longrightarrow L^2(\R_+)$ is such a unitary
operator, i.e. $Ul_D= l_NU$, then we have
\bead
\lefteqn{
A^N = l_N\otimes I_{\cH} + I_{\gH}\otimes T = }\\
& &
(U\otimes I_{\cH})[l_D\otimes I_{\cH} + I_{\gH}\otimes T](U^*\otimes I_{\cH})
= (U\otimes I_{\cH})A^D(U^*\otimes I_{\cH}).
\eead
The proof can be extended to any non-negative realization $l_h$ of
$l$ fixed by the domain $\dom(l_h)=\{f\in W^{1,2}(\R_+): f'(0) = hf(0),\quad h \ge 0\}$.
Moreover, a proof of the absolutely continuity of $A^D$ and $A^N$,
which does not used boundary triplets,
can be found in Appendix \ref{A.2}. For the Krein realization $A^K$ we do not know such proofs.
}
\er

Next we describe the spectral properties of any self-adjoint
extension of $A$. In particular, we show that the
Friedrichs extension $A^F$ of $A$ is $ac$-minimal, though $A$ does
not satisfy conditions of Theorem  \ref{VI.6}.
\bt\la{VII.3}
Let $T\ge 0$, {$T = T^*\in [\cH]$}, and $t_1:=\inf\gs_{\ess}(T)$.
Let also $A$ be the symmetric operator defined by \eqref{8.1} and
$\wt A = \wt A^* \in \Ext_A$.  Then

\item[\rm\;\;(i)] the absolutely continuous part $\wt A^{ac}
E_{\wt A}([t_1,\infty))$ is
unitarily equivalent to the part $A^DE_{A^D}([t_1,\infty))$;

\item[\rm\;\;(ii)] the Dirichlet, Neumann and Krein realizations are
$ac$-minimal and $\gs(A^D) = \gs(A^N) = \gs_{ac}(A^K) \subseteq
\sigma_{ac}(\wt A)$;

\item[\rm \;\;(iii)] the absolutely continuous part $\wt A^{ac}$
is unitarily equivalent to $A^D$  whenever either $(\wt
A - i)^{-1} - (A^D - i)^{-1} \in \gotS_\infty(\gotH)$ or $(\wt A -
i)^{-1} - (A^K - i)^{-1} \in \gotS_\infty(\gotH)$.
\et
\begin{proof}
By \cite[Corollary 4.2]{MN2011} it  suffices to assume that the
extension $\wt A = \wt A^*$  is disjoint with $A_0,$ that is, by
Proposition \ref{prop2.1}(ii) it admits a representation  $\wt A
= A_B$ with $B\in \cC(\cH).$

(i) Let $\Pi=\{\cH, \gG_0,\gG_1\}$ be a boundary triplet for $A^*$
defined by  \eqref{8.2}.  In accordance with Theorem \ref{IV.10}
we calculate $d_{M_B^K}(t)$ where $M_B(\cdot) :=
(B-M(\cdot))^{-1}$ is the generalized Weyl function of
the extension $A_B$, cf. \eqref{2.5}. Clearly,
\be\label{7.15}
 \im(M_B(z)) = M_B(z)^*\im(M(z))M_B(z), \quad z \in
\C_+. \ee
Since  $\re(\sqrt{z - \gl})>0$ for $z = t +iy,\  y>0,$  it
follows from \eqref{8.2B} that
      \be\label{7.16}
  \im(M(z)) = \int_{[0,\infty)}\re(\sqrt{z -
\gl})\;dE_T(\gl) \ge \int_{[0,\tau)}\re(\sqrt{z -
\gl})\;dE_T(\gl),
      \ee
where $z = t +iy$. It is easily seen that
      \be \label{7.16A}
\re(\sqrt{z - \gl}) \ge \sqrt{t-\gl} \ge
\sqrt{t-\tau}, \quad \gl \in [0,\tau), \quad t > \tau.
     \ee
Combining \eqref{7.15} with \eqref{7.16} and \eqref{7.16A} we get
\bed  \im(M_B(t+iy)) \ge
\sqrt{t-\tau}M_B(t+iy)^*E_T([0,\tau))M_B(t+iy), \quad t > \tau
>0. \eed
Let $Q$ be a finite-dimensional orthogonal projection, $Q \le
E_T([0,\tau))$. Hence
\bed \im(M_B(t+iy)) \ge \sqrt{t-\tau}M_B(t+iy)^*QM_B(t+iy), \quad
t > \tau >0, \quad y > 0. \eed
Setting $\cH_1= \ran(Q)$, $\cH_2 := \ran(Q^\perp)$, and
choosing $K_2\in \mathfrak S_2(\cH_2)$ and  satisfying $\ker(K_2)
= \ker(K^*_2) = \{0\}$, we define a Hilbert-Schmidt operator $K :=
Q \oplus K_2\in \mathfrak S_2(\cH).$  Clearly,  $\ker(K) =
  \ker(K^*) = \{0\}$ and,
\bea\label{6.14A}
\lefteqn{
\im(K^*M_B(t+iy)K) \ge }\\
& &
\sqrt{t-\tau}K^*M_B(t+iy)^*QM_B(t+iy)K, \quad t > \tau >0.
\nonumber
\eea
Since  $M_B(\cdot)\in (R_\cH)$ and   $Q,\ K\in \mathfrak
S_2(\cH),$  the limits
\bead
K^*M_B(t)^*Q & := & \slim_{y\downarrow 0}K^*M_B(t + iy)^*Q \quad \mbox{and} \quad \\
(QM_BK)(t)   & := & \slim_{y\downarrow 0}QM_B(t+iy)K
\eead
exist for a.e. $t\in \R$ (see \cite{BirEnt67}). Therefore
passing to the limit as $y\to 0$ in \eqref{6.14A}, we arrive at
the inequality
\bed
\im(M_B^K(t)) \ge \sqrt{t-\tau}(K^*M_B(t)^*Q)(QM_BK(t)), \quad  t
> \tau >0, \quad y> 0.
\eed
It follows that
\be\label{7.18}
\dim(\ran\left((QM_BK)(t)\right)) \le
\dim(\ran\left(\im M_B^K(t)\right)) =d_{M^K_B}(t), \quad t > \tau.
\ee
We set $\wt M^Q_B(z) :=QM_B(z)Q \upharpoonright\cH_1.$
Since $\dim(\cH_1)<\infty$ the limit $\wt M^Q_B(t):=
\slim_{y\downarrow 0}\wt M^Q_B(t+iy)$ exists  for a.e. $t\in \R.$
Since $(QM_BK)(t)\upharpoonright\cH_1 = \ran\left((\wt
M^Q_B)(t)\right)$, \eqref{7.18} yields the inequality
        \be\la{7.18a}
\dim(\ran\left(\wt M^Q_B(t)\right)) \le
\dim(\ran\left((QM_BK)(t)\right)) \le d_{M^K_B}(t)
         \ee
for a.e. $t \in [\tau,  \infty)$.

Since  $\dim(\cH_1) <\infty$ and $\ker(\wt M^Q_B(z)) =\{0\},
z\in \C,$ we easily get by  repeating the corresponding reasonings
of the proof of Theorem   \ref{VI.6} that $\ran\left(\wt
M^Q_B(t)\right) = \cH_1$ for a.e. $t \in \R$. Therefore
\eqref{7.18a} yields  $\dim(\cH_1) \le d_{M^K_B}(t)$ for a.e. $t
\in [\tau, \infty).$

If $\tau > t_1$, then $\dim(E_T([0,\tau))\cH) =\infty$ and
the dimension of a projection $Q \le E_T([0,\tau))$ can be
arbitrary. Thus, $d_{M_B^K}(t) = \infty$ for a.e. $t >
\tau$. Since $\tau > t_1$ is arbitrary we get
$d_{M^K_B}(t) = \infty$ for a.e. $t
> t_1$. By  Theorem \ref{IV.10}(ii)
the operator $\wt A^{ac}E_{\wt A}([t_1,\infty))$ is
unitarily equivalent to $A_0E_{A_0}([t_1,\infty))$.

(ii) If $\tau \in (t_0, t_1)$, then $\dim(E_T([0,\tau))\cH) =:
p(\tau)<\infty.$  Hence, $\dim(Q\cH) \le p(\tau)$ which shows that
$d_{M_B^K}(t) \ge p(\tau)$ for a.e. $t \in (\tau, t_1)$. Since
$\tau$ is arbitrary, we obtain $d_{M_B^K}(t) \ge p(\tau)$ for
a.e. $t \in [0,t_1)$. Using Theorem \ref{IV.10}(i) we prove
$A^D$ is $ac$-minimal. Using Theorem \ref{prop7.2A}(iv) we complete the proof
of (ii).

(iii) By Lemma \ref{VII.1} the invariant maximal normal
function $\gotm^+(t)$ is finite for $t \in \R$. By Theorem \ref{V.5}
$\wt A^{ac}$ and $(A^F)^{ac}$ are unitarily equivalent.
Similarly we prove that $\wt A^{ac}$ and $(A^K)^{ac}$ are
unitarily equivalent. To complete the proof it remains to
apply Theorem \ref{prop7.2A}(i).
\end{proof}

Using Definition \ref{III.5} one gets the following corollary.
\bc\la{VI.5Z}
Let the assumptions of Theorem \ref{VII.3} be satisfied. If
$\dim(\cH) = \infty$ and $t_0:=\inf\gs(T) =
\inf\gs_{\ess}(T)=:t_1$, then

\item[\;\;\rm (i)]  the Dirichlet, Neumann and Krein realizations are strictly
$ac$-minimal;
\item[\;\;\rm (ii)] the absolutely continuous part $\wt A^{ac}$ of $\wt A$ is
unitarily equivalent to $A^D$,  whenever
\be\label{4.18A}
(\wt A - i)^{-1} - (A^N - i)^{-1} \in \gotS_\infty(\gotH).
\ee
\ec
\begin{proof}
(i) This statement follows from Theorem \ref{VII.3}(i) and Theorem \ref{prop7.2A}.

(ii) To prove this  statement we note that by the Weyl theorem the
inclusion \eqref{4.18A}  yields $\sigma_{\ess}(\wt
A)=\sigma_{\ess}(A^N)$. Since $\sigma_{\ess}(A^N)=
\sigma_{ac}(A^N)=[t_0,\infty)$ we have $\sigma_{\ess}(\wt
A)=[t_0,\infty)$.  On the other hand, by Theorem \ref{VII.3}(i) we get
$[t_0,\infty)=\gs_{\ess}(\wt A)\ \subseteq \gs_{ac}(\wt A)$. Thus,
$\sigma_{ac}(\wt A)=[t_0,\infty)$ and $\wt
A^{ac}=\wt A^{ac} E_{\wt A}\bigl([t_0,\infty)\bigr)$. Using
Theorem \ref{prop7.2A}(i) and again Theorem
\ref{VII.3}(i) we find that $\wt A^{ac}$ is unitarily equivalent to $A^D$.
\end{proof}
\br
{\em
According to \eqref{7.14B} the condition $\gotm^+(t)<\infty$,
$t\in \R$ (cf. \eqref{4.12A}) is not satisfied for the  Weyl
function $M_0(\cdot)$ of the Neumann extension $A^N$. Thus, the
statement (ii) of Corollary \ref{VI.5Z}  shows that the assumption
$\gotm^+(t)<\infty$ of Theorem \ref{V.5}, which is a generalization of the
classical Kato-Rosenblum theorem, is sufficient but not necessary for validity of
the conclusions.
}
\er
\bc\la{VI.5Zz}
Let the assumptions of Theorem \ref{VII.3} be satisfied and let
$\dim(\cH) = \infty$. Then $A^D$ is strictly $ac$-minimal if and only if $t_0 =
t_1$.
\ec
\begin{proof}
Let $t_0 < t_1$. Then there is a decomposition $T = T_{\rm fin} \oplus
T_\infty$ such that $T_{\rm fin}$ acts in a finite dimensional Hilbert space
$\cH_{\rm fin}$ and $t_0 = \inf\gs(T_{\rm fin})$ and $T_\infty =
T^*_\infty \in \cC(\cH_\infty)$ and
$t_0 < t_\infty := \inf\gs(T_\infty) \le t_1$. This leads to the decomposition
$A = A_{\rm fin} \oplus A_\infty$ where $A_{\rm fin}$ and $A_\infty$ are defined analogously to
\eqref{8.1}. Clearly $A^D = A^D_{\rm fin} \oplus A^D_\infty$. By
Theorem \ref{prop7.2A} both extensions $A^D_{\rm fin}$ and $A^D_\infty$
are absolutely continuous and their spectra are given by $\gs(A^D_{\rm fin}) =
[t_0,\infty)$ and $\gs(A^D_\infty) = [t_\infty,\infty)$.
Since $\dim(\cH_\infty) = \infty$ the deficiency indices of
$A_\infty$ are infinite. We note that $(-\infty,t_\infty)$ is a
spectral gap for $A_\infty$. Using a result of Brasche \cite{Bra04}
there exists an extension $\wt A_\infty = \wt A^*_\infty \in \Ext{A_\infty}$ such that
$\gs(\wt A_\infty) \subseteq [t_0,\infty)$, the
part $\wt A_\infty E_{\wt A_\infty}([t_0,t_\infty))$ is absolutely
continuous and $N_{\wt A^{ac}_\infty}(t) = \infty$ for $t \in
[t_0,t_1)$.

Let $\wt A := A^D_{\rm fin} \oplus \wt A_\infty$. The
operator $\wt A$ is a self-adjoint extension of $A$ such that
$\gs(\wt A) = \gs(A^D) = [t_0,\infty)$. The parts
$A^DE_{A^D}([t_0,t_\infty))$ and $\wt AE_{\wt A}([t_0,t_\infty))$
are absolutely continuous. However, the absolutely continuous parts of
both extensions are not unitarily equivalent. Indeed, for a.e. $t \in
[t_0,t_\infty)$ one has $N_{A^D}(t) < \infty$ but
$N_{\wt A^{ac}}(t) = \infty$, by construction. Hence $A^D$ is not
strictly $ac$-minimal which yields $t_0 = t_1$. The converse follows
from Corollary \ref{VI.5Z}(i).
\end{proof}

\section{Sturm-Liouville operators with unbounded operator potentials}

\subsection{{Regularity properties}}

In this subsection we  consider the differential expression
\eqref{8.1} with unbounded non-negative  $T = T^* ( \in \cC(\cH))$
in $\gH:=L^2({\R}_+,\cH)$. The minimal operator $A :=A_{\rm \min}
:= \overline{\cA}$, cf.\eqref{8.10} and \eqref{8.11}, is densely
defined and non-negative.   If $T$ is bounded, then $A$ coincides
with \eqref{8.1}.

Let $\cH_1(T)$ be the Hilbert space which is obtained
equipping the set $\dom(T)$ with the graph norm of $T$. Moreover,
for any $s\ge 0$ we equip  $\dom(T^s)$ with the graph norm
\be\la{5.1a}
\|u\|_s=(\|u\|^2_{\cH}+\|T^{{s}} u\|^2_{\cH})^{1/2},
\qquad s \ge 0, \quad u \in \cH,
\ee
and denote  by $\cH_s(T)$ the corresponding the Hilbert space.
{Following \cite[Definition I.2.1]{LioMag72} the intermediate spaces
$[X,Y]_\gth$, $\gth \in [0,1]$,  of $X = \cH_1(T)$ and $Y = \cH_0(T) := \cH$ are defined
by $[X,Y]_\gth = \cH_{1-\gth}(T)$, $\gth \in [0,1]$.}

Furthermore, by $\cH_s(T)$, $s < 0$, we denote the completion of
$\cH$ with respect to the "negative" norm
\be\la{5.1b}
\|u\|s =\|(I+T^{-2s})^{-1/2}u\|_{\cH},
\qquad  s < 0, \quad  u\in\cH.
\ee
At first, we describe the domain $\dom(A)$ of the minimal operator
$A$. For this purpose,  following \cite{LioMag72} we introduce the
Hilbert spaces $W^{k,2}_T(\R_+,\cH) := W^{k,2}(\R_+,\cH) \cap
L^2(\R_+,\cH_1(T))$, $k \in \N$, equipped with the Hilbert norms
\bed
\|f\|^2_{W^{k,2}_T} =\int_{{\R}_+}\bigl(\|f^{(k)}(t)\|^2_{\cH} +
\|f(t)\|^2_\cH +\|T f(t)\|_{\cH}^2\bigr)dt.
\eed
Obviously we have $\cD_0 \subseteq W^{2,2}_T({\R}_+,\cH)$ where
is given by \eqref{8.11}. The closure of $\cD_0$ in
$W^{2,2}_T({\R}_+,\cH)$ coincides with $W^{2,2}_{0,T}(\R_+,\cH) := \{f\in
W^{2,2}_T({\R}_+,\cH):\  f(0)=f'(0) =0\}$ which yields
$W^{2,2}_{0,T}(\R_+,\cH) \subseteq \dom(A)$.
\bl\la{VI.2.7}
Let $T=T^*$ be a non-negative  operator in $\cH$. Then the domain
$\dom(A)$ equipped with the graph norm coincides with  the Hilbert
space $W^{2,2}_{0,T}(\R_+,\cH)$  algebraically and topologically.
\el
\begin{proof}
Obviously, for any $f \in \cD_0$ we have
\bead
\lefteqn{
\left\|\cA f\right\|^2_\gotH
= \int_{\R_+}\left\|f''(x)\right\|^2_\cH dx } \\
     & &
 + \int_{\R_+}\|Tf(x)\|_\cH ^2dx - 2\real\left\{\int_{\R_+}
\left(f''(x),Tf(x)\right)_\cH dx\right\}.
           \eead
Integrating by parts we find
         \bed
\int_{\R_+} \left(f''(x),Tf(x)\right)dx =
-\int_{\R_+}\left\|\sqrt{T}f'(x)\right\|^2_\cH dx.
        \eed
Hence
\bed
\left\|\cA f\right\|^2_\gotH =
\int_{\R_+}\left\|f''(x)\right\|^2dx +
  \int_{\R_+}\|Tf(x)\|^2dx +
2\int_{\R_+}\left\|\sqrt{T}f'(x)\right\|^2_\cH dx
\eed
for any $f \in \cD_0$ which yields
\bed
\|f\|^2_{W^{2,2}_T} \le \|\cA f\|^2_\gotH + \|f\|^2, \quad f \in
\cD_0.
\eed
Furthermore, by the Schwartz inequality,
\bed
2\left|\real\left\{\int_{\R_+} \left(f'(x),Tf(x)\right)_\cH dx\right\}\right| \le
\|f\|^2_{W^{2,2}_T}, \quad f \in \cD_0.
\eed
which gives
\bed
\|\cA f\|^2_\gotH + \|f\|^2 \le 2\|f\|^2_{W^{2,2}_T}, \quad f \in \cD_0.
\eed
Thus, we arrive at the two-sided  estimate
\bed
\|f\|^2_{W^{2,2}_T} \le \left\|\cA f\right\|^2_\gotH +
\|f\|^2_\gotH \le 2\|f\|^2_{W^{2,2}_T}, \qquad f \in \cD_0.
\eed
Since  $\cD_0$ is dense in $W^{2,2}_{0,T}(\R_+,\cH)$ we obtain
that  $\dom(A)$ coincides with $W^{2,2}_{0,T}(\R_+,\cH)$
algebraically and topologically.
\end{proof}

In opposite to the case of the minimal operator $A=A_{\min}$
the maximal operator $A_{\max}= A_{\min}^*$  obviously satisfies
$W^{2,2}_T(\R_+,\cH)\subset\dom(A_{\max})$,  though
$\dom(A_{\max})\not = W^{2,2}_T(\R_+,\cH)$ if $T$ is not
bounded. Moreover, it was firstly shown in \cite{Gor71} (see also
\cite[Section 4.1]{GG91}) that the trace mapping
\bed
\{\gamma_0,\gamma_1\}:\ W^{2,2}_T(I,\cH) \longrightarrow \cH_{3/4}(T)\oplus\cH_{1/4}(T),\quad
\{\gamma_0,\gamma_1\}f=\{f(a),f'(a)\},
\eed
can be extended to a continuous (non-surjective) mapping
\bed
\{\gamma_0,\gamma_1\}:\
\dom(A_{\max})\to\cH_{-1/4}(T)\oplus\cH_{-3/4}(T).
\eed
It is also shown in \cite[Theorem 4.1.1]{GG91} that
$y(\cdot)\in\dom(A_{\max})$ if and only if the following conditions
are satisfied:
\begin{enumerate}
\item[\;\;(i)] $y'(\cdot)$ exists and is an absolutely continuous function on
$I$ into $\cH_{-1}(T)$;
\item[\;\;(ii)]
$\cA y\in L^2(I,\cH)$.
\end{enumerate}
This result is similar to that for elliptic operators with smooth
coefficients in domains with smooth boundary,
cf. \cite{Gru08,LioMag61}. A similar statement holds also for the operator
$A_{\max}= A_{\min}^*$  considered in $L^2(\R_+,\cH)$, cf.
\cite[Section 9]{DM91}.

Next, we investigate the Friedrichs extension $A^F$ and the Krein extension $A^K$ of
the operator $A\ge0$. We define also the Neumann realization $A^N$
as the self-adjoint operator associated with the closed quadratic
form ${\mathfrak t}_N$,
\be\la{4.18}
{\mathfrak t}_N[f]  := \int^\infty_0 \left\{\|f'(x)\|^2_\cH +
\|{\sqrt T}f(x)\|^2_\cH\right\}dx = \|f\|^2_{W^{1,2}_{\sqrt T}} - \|f\|^2_{L^2(\R_+,\cH)},
\ee
$f \in \dom({\mathfrak t}_N)  :=   W^{1,2}_{\sqrt T}(\R_+,\cH)$.
Clearly, $A^N \in \Ext_A.$  In the case of bounded  $T$ one has
$A^N=A_1$ where $A_1$ is defined in  Theorem \ref{prop7.2A}(ii).

We note that the closed quadratic $\mathfrak t_F$ associated with
Friedrich extensions $A^F$ is given by $\gt_F := \gt_N\upharpoonright\dom(\gt_F)$,
$\dom(\gt_F) := \{f\in W^{1,2}_{\sqrt{T}}(\R_+,\cH):\ f(0)=0\}$.
\bp\label{prop6.8}
Let $T=T^*\in\cC(\cH)$, $T\ge 0$, and let $A:=\overline{\cA}$
Let also $\cH_n:=\ran\bigl(E_T([n-1,n))\bigr)$, $T_n:= T E_T([n-1,n))$, $n\in\N$, and
let $S_n$ be the closed minimal symmetric operator defined
by \eqref{8.1} in $\gH_n:=L^2({\mathbb R}_+,\cH_n)$ with $T$
replaced by $T_n.$  Then
\item[\rm\;\;(i)] the following decompositions hold
\be\label{8.23}
 A = \bigoplus^\infty_{n=1} S_n, \quad A^F =
\bigoplus^\infty_{n=1} S^F_n, \quad A^K = \bigoplus^\infty_{n=1}
S^K_n, \quad A^N = \bigoplus^\infty_{n=1}S^N_n;
\ee
\item[\rm\;\;(ii)] the domain $\dom(A^F)$ equipped with the
graph norm is a closed subspace of $W^{2,2}_T(\R_+,\cH)$ is given by
$\dom(A^F) = \{f \in W^{2,2}_T(\R_+,\cH):  f(0) = 0\}$;

\item[\rm\;\;(iii)] the domain $\dom(A^N)$ equipped with the
graph norm is a closed subspace of $W^{2,2}_T(\R_+,\cH)$, is give by
$\dom(A^N) = \{f \in W^{2,2}_T(\R_+,\cH): f'(0)=0\}$.
\end{proposition}
\begin{proof}
(i)
{Since Lemma \ref{VI.2.7} is valid for bounded $T$ we find that
the graph $\graph(S_n)$ of $S_n$
equipped with usual graph norm is algebraically and
topologically equivalent to $W^{2,2}_{T_n}(\R_+,\cH_n)$, $n \in \N$.
Obviously, we have
\bed
W^{2,2}_T(\R_+,\cH) = \bigoplus_{n\in\N}W^{2,2}_{T_n}(\R_+,\cH_n)
\eed
which yields
\bed
\graph(A)  = \bigoplus_{n\in\N} \graph(S_n).
\eed
However, the last relation proves the first relation of \eqref{8.23}.}

The second and the third relations are implied by  Corollary \ref{cor5.5}.
To prove the last relation of \eqref{8.23} we set $S^N :=
\bigoplus^\infty_{n=1}S^N_n.$ Since $S^N_n= (S^N_n)^*\in
\Ext_{S_n}$ and $A = \bigoplus^\infty_{n=1} S_n,$  $S^N$ is a
self-adjoint extension of $A,\ S^N\in \Ext_A$. Let
$f=\bigoplus^\infty_{n=1} f_n\in \gotH$ where
$\gotH = \bigoplus^\infty_{n=1}\gotH_n$. Denoting by
$\wt{{\mathfrak t}}_N$ the quadratic form associated with $S^N$
we find $f=\bigoplus^\infty_{n=1} f_n\in \dom(\wt{{\mathfrak t}}_N)$
if and only if $f_n \in \dom(\mathfrak t_n)$, $n \in \N$, and
$\sum^\infty_{n=1} \mathfrak t_n[f_n] < \infty$ where
$\mathfrak t_n$ is the quadratic form associated with $S^N_n$,
$n \in \N$. If $f \in \dom(\wt{{\mathfrak t}}_N)$, then
\bed
\begin{aligned}
\wt{{\mathfrak t}}_N[f] = \sum^\infty_{n=1}\mathfrak t_n[f_n] =
\sum^{\infty}_{n=1}& \int^\infty_0 \left\{\|f'_n(x)\|^2_{\cH_n}  +
  \|{\sqrt T_n}f_n(x)\|^2_{\cH_n}\right\}dx \\
= \int^\infty_0 &\left\{\|f'(x)\|^2_{\cH} + \|{\sqrt T}f(x)\|^2_{\cH}\right\}dx
= {\mathfrak t}_N[f]
\end{aligned}
\eed
which yields $f \in \dom({\mathfrak t}_N)$. Conversely, if
$f \in \dom({\mathfrak t}_N)$ and $f = \bigoplus^\infty_{n=1}f_n$,
then $f_n \in \dom(\mathfrak t_n)$, $n \in \N$, and
$\sum^\infty_{n=1}\mathfrak t_n[f_n] < \infty$ which proves
$f \in \dom(\wt{{\mathfrak t}}_N)$. Hence $S^N = A^N$.

(ii) Following the reasoning of Lemma \ref{VI.2.7} we find
 \be\la{8.23b}
\|f_n\|^2_{W^{2,2}_{T_n}} \le \|S^F_n f_n\|^2_{\gotH_n} +
\|f_n\|^2_{\gotH_n} \le 2\|f_n\|^2_{W^{2,2}_{T_n}}, \qquad n \in \N,
\ee
where  $f_n \in \dom(S^F_n) = \{g_n \in
W^{2,2}(\R_+,\cH_n): g_n(0) =0\}$. { Using
 representation \eqref{8.23} for $A^F$ and setting $f^m :=
\bigoplus^m_{n=1} f_n$, $f_n \in \dom(F_n)$, we obtain from
\eqref{8.23b}
\be\la{8.23a}
 \|f^m\|^2_{W^{2,2}_T} \le \|A^F f^m\|^2_{\gotH} +
\|f^m\|^2_{\gotH} \le 2\|f^m\|^2_{W^{2,2}_T}, \quad m \in \N.
\ee
{ Since the set  $\{f^m = \bigoplus^m_{n=1} f_n: \ f_n \in
\dom(S^F_n),\ m \in \N \}$, is  a core for $A^F,$ inequality
\eqref{8.23a} remains valid for} $f \in \dom(A^F)$. This shows
that $\dom(A^F) = \{f \in W^{2,2}_T(\R_+,\cH): f(0) = 0\}$. {
Moreover, due to \eqref{8.23a}  the graph norm of $A^F$ and  the
norm $\|\cdot\|_{W^{2,2}_T}$ restricted to $\dom(A^F)$ are
equivalent}.

(iii) Similarly to  \eqref{8.23b} one gets}
     \bed
\|f_n\|^2_{W^{2,2}_{T_n}} \le \|S^N_n f_n\|^2_{\gotH_n} + \|f_n\|^2 \le 2\|f_n\|^2_{W^{2,2}_{T_n}}
         \eed
for $f_n \in \dom(S^N_n) = \{g_n \in W^{2,2}(\R_+,\cH_n): g'_n(0)
= 0\}$, $n \in \N$. { It remains to repeat the reasonings  of
(ii)}.
\end{proof}

In the following we denote by $C_b(\R_+,\cH_s)$, $s \in [0,1]$,
the space of bounded continuous functions $f : \R_+ \longrightarrow
\cH_s$.
\bc\label{cor5.4}
Let the assumptions of Proposition \ref{prop6.8} be satisfied. Further,
let $\partial f := f'$ be the derivative of $f
\in W^{2,2}(\R_+,\cH)$ in the distribution sense. If $f \in \dom(A^D) \cup \dom(A^N)$, then

\item[\;\;\rm (i)]
$\partial f := f' \in L^2(\R_+,\cH_{1/2}(T))$ and the maps
\bead
\partial : & & \dom(A^D) \ni f \longrightarrow f' \in L^2(\R_+,\cH_{1/2}(T)),\\
\partial : & & \dom(A^N) \ni f \longrightarrow f' \in L^2(\R_+,\cH_{1/2}(T))
\eead
are continuous;

\item[\;\;\rm (ii)]
$f(\cdot)\in C_b(\R_+,\cH_{3/4}(T))$, $f'(\cdot)\in C_b(R_+,\cH_{1/2}(T))$ and the maps
\bead
\partial^j : & & \dom(A^D) \ni f \longrightarrow f^{(j)} \in C_b(\R_+,\cH_{3/4- j/2}(T)),\\
\partial^j : & & \dom(A^N) \ni f \longrightarrow f^{(j)} \in C_b(\R_+,\cH_{3/4- j/2}(T)),
\eead
$j = 0,1$, are continuous. In particular, one has $f(0)\in \cH_{3/4}(T)$ and
$f'(0)\in\cH_{1/4}(T)$.
\ec
\begin{proof}
(i) {From Proposition \ref{prop6.8}(ii) and (iii) we get that
$u \in L^2(\R_+,X)$, $X = \cH_1(T)$. Applying the intermediate Theorem I.2.3
of \cite{LioMag72} to $X \subseteq Y = \cH_0 := \cH$ we immediately
obtain $f' \in L^2(\R_+,[X,Y]_{1/2})$ which yields $f' \in
L^2(\R_+,\cH_{1/2}(T))$. Moreover, it follows that the map $\partial$
is continuous.
}

(ii) Combining Proposition \ref{prop6.8}(ii) and (iii) with the trace
theorem \cite[Theorem 1.3.1]{LioMag72} one proves (ii).
\end{proof}
\begin{remark}\la{V.4}
{\em
Lemma \ref{VI.2.7},  Proposition \ref{prop6.8} and Corollary
\ref{cor5.4} also hold for realizations of the differential
expression $\cA$ considered on a finite interval $I$, i,e, in
the space $L^2(I,\cH)$. For this case Corollary \ref{cor5.4}
has firstly been proved by M.L. Gorbachuk \cite{Gor71} (see also
\cite[Corollary 4.1.5]{GG91}, \cite[Theorem 4.2.4]{GG91}) by
applying another method. Realizations $\wt A\in \Ext_A$ satisfying
the condition $\dom(\wt A)\subset C(I,\cH_{3/4}(T))$ are called
maximally smooth (see \cite[Section 4.2]{GG91}).

We emphasize however, that Lemma \ref{VI.2.7} and  Proposition
\ref{prop6.8} are new for the case of finite interval realizations
too.
}
\end{remark}

\subsection{Operators on semi-axis: Spectral properties.}

To extend Theorem \ref{prop7.2A} to the case of
unbounded operators $T=T^*\ge0$ we firstly construct a boundary
triplet for $A^*$, using Theorem \ref{VI.3} and
representation \eqref{8.23} for $A$.
\bl\label{lem6.9}
Let the assumptions of Proposition \ref{prop6.8} be satisfied. Then there is a
sequence of boundary triplets $\wh\gP_n =
\{\cH_n,\wh\gG_{0n},\wh\gG_{1n}\}$ for $S_n^*$ such that
$\gP := \bigoplus^\infty_{n=1}\wh\gP_n =: \{\cH,\wh
  \gG_0,\wh \gG_1\}$ forms a boundary triplet for $A^*.$
Moreover,  $A^F = A^*\upharpoonright\ker(\wh \gG_0)$ and
the corresponding  Weyl function is given by
\be\la{8.7}
\wh M(z) = \frac{i\sqrt{z-T} +\im(\sqrt{i - T})}{\re(\sqrt{i-T})}.
\qquad z \in \C_+,
\ee
 \el
\begin{proof}
For any $n \in \N$ we choose a boundary triplet $\gP_n =
\{\cH_n,\gG_{0n},\gG_{1n}\}$ for $S_n^*$ with
$\gG_{0n},\gG_{1n}$ defined by \eqref{8.2}. By Theorem
\ref{prop7.2A}(i) $S^F_n = S_{0n} =
S_n^*\upharpoonright\ker(\gG_{0n})$ and by Lemma
\ref{lem6.1}
the corresponding Weyl function is $M_n(z) = i\sqrt{z-T_n}$.

Following Lemma \ref{VI.1}, cf. \eqref{3.8}, we  define a sequence
of regularized boundary triplets $\wh\gP_n =
\{\cH_n,\wh\gG_{0n},\wh\gG_{1n}\}$ { for $S_n^*$} by setting
$R_n:=(\re(\sqrt{i-T_n}))^{1/2},$ \ $Q_n:= -\im(\sqrt{i - T_n})$
and
\be\label{6.22A}
\wh\gG_{0n} := R_n\gG_{0n}, \quad \wh\gG_{1n} := R_n^{-1}(\gG_{1n}
- Q_n\gG_{0n}), \qquad n\in \N.
\ee
Hence { $S^F_n = S_{0n}$} and the corresponding  Weyl
function $\wh M_n(\cdot)$ is given by
\be \la{8.7B}
\wh M_n(z) = \frac{i\sqrt{z-T_n}
  +\im(\sqrt{i - T_n})}{\re(\sqrt{i-T_n})},
\qquad z \in \C_+, \qquad n \in \N.
\ee
By  Theorem \ref{VI.3} the direct sum $\wh\gP :=
\bigoplus_{n=1}^\infty\wh\gP_n  = \{\cH,\wh\gG_{0},\wh\gG_{1}\}$
forms a boundary triplet for $A^*$ and
the corresponding Weyl function is
\be\la{8.7C}
 \wh M(z) = \bigoplus_{n\in\N}\wh M_n(z), \qquad z \in \C_+.
\ee
Combining \eqref{8.7C} with  \eqref{8.7B} we arrive at \eqref{8.7}.
From Theorem \ref{VI.3} (cf. \eqref{6.8})  and
Corollary \ref{cor5.5} we get
\be\la{6.25z}
A_0 = A^*\upharpoonright\ker(\wh \gG_0) =
\bigoplus^\infty_{n=1}S^*_n\upharpoonright\ker(\wh \gG_{0n}) =
\bigoplus^\infty_{n=1}S_{0n} =  \bigoplus^\infty_{n=1}S^F_{n}=A^F
\ee
which proves the second assertion.
\end{proof}

Next we generalize Theorem \ref{prop7.2A} to the case of
unbounded operator potentials.
\bt\la{VI.9}
Let $T = T^* \ge 0,$\  $t_0 := \inf\gs(T)$. Let $A := A_{\rm min}$
be the minimal operator associated with $\cA$, cf. \eqref{8.10}
and let $\wh \gP = \{\cH,\wh \gG_0,\wh \gG_1\}$ be the boundary
triplet for $A^*$ defined by Lemma \ref{lem6.9}.
Then the following holds:

\item[\rm \;\;(i)]
The Dirichlet realization $A^Df := \cA f$, $f \in \dom(A^D) :=
\{g \in W^{2,2}_T(\R_+,\cH): g(0) = 0\}$ coincides with
$A_0 := A^*\upharpoonright\ker(\wh \gG_0)$ which is identical
with the Friedrichs extension $A^F$.
Moreover, $A^D$ is absolutely continuous and $\gs(A^D) = \gs_{ac}(A^D) = [t_0,\infty)$.

\item[\rm \;\;(ii)]
The  Neumann realization $A^N := \cA f$, $f \in \dom(A^N) := \{g \in
W^{2,2}_T(\R_+,\cH: g'(0) = 0\}$ coincides with $A_{B^N} :=
A^*\upharpoonright\dom(A_{B^N})$ where $\dom(A_{B^N}) =
\dom(\ker(\wh \gG_1 - B^N\wh \gG_0))$ and $B^N := \sqrt{T + \sqrt{I + T^2}}$.
Moreover, $A^N$ is absolutely continuous
$\gs(A^N) = \gs_{ac}(A^N) = [t_0,\infty)$.

\item[\rm \;\;(iii)]
The Krein realization (or extension) $A^K$ is given by
$A_{B^K}:=A^*\upharpoonright \ker(\gG_1 - B^K\gG_0)$, where
\be\la{6.32x}
B^K = \frac{1}{\sqrt{2}\sqrt{T} + \sqrt{T + \sqrt{1
+ T^2}}} \frac{1}{\sqrt{T + \sqrt{1 + T^2}}}.
\ee
Moreover, $\ker(A^K) = \gotH_0 := \overline{\gotH'_0}$, $\gotH'_0
:= \{e^{-x\sqrt{T}}h: h \in \ran(T^{1/4})\},$ the restriction $A^K
\upharpoonright\dom(A^K) \cap \gotH^\perp_0$ is absolutely
continuous, and $A^K = 0_{\gotH_0} \bigoplus (A^K)^{ac}$.
In particular, $\gs(A^K) = \{0\} \cup \gs_{ac}(A^K)$ and
$\gs_{ac}(A^K) = [t_0,\infty)$.

\item[\rm \;\;(iv)]
The realizations $A^D$, $A^N$ and $(A^K)^{ac}$ are unitarily
equivalent.
\et
\begin{proof}
(i) From Proposition \ref{prop6.8}(ii) we get $A^D = A^F$. Applying Lemma \ref{lem6.9} we get
$A^F = A_0$. Finally, using Proposition \ref{prop6.8}(i) and
Theorem \ref{prop7.2A}(i) we verify the remaining part.

(ii) It is easily seen that with respect to the boundary triplet $\wh \gP_n
= \{\cH_n,\wh \gG_{0n},\wh \gG_{1n}\}$ defined by
\eqref{6.22A} the  extension $A^N_n$ admits a representation
$A^N_n = A_{B_n}$ where  $B_n := \sqrt{T_n + \sqrt{1 + T^2_n}}$,
$n \in \N$. By Proposition \ref{prop6.8}(i),  $A^N =
\bigoplus^\infty_{n=1}A^N_n = A_{B^N}$ where $B^N =
\bigoplus^\infty_{n=1}B_n$. The remaining part of (ii) follows from
the representation $A^N = \bigoplus^\infty_{n=1}A^N_n$ and Theorem
\ref{prop7.2A}(ii).

(iii) Using  the polar decomposition $i-\gl = \sqrt{1 +
\gl^2}e^{i\theta(\gl)}$ with $\theta(\gl) = \pi - \arctan(1/\gl)$,
$\gl \ge 0$ we get
\be \label{6.23A} \re(\sqrt{i - T}) = \int^\infty_0 \sqrt[4]{1 +
\gl^2}\cos(\theta(\gl)/2)dE_T(\gl). \ee
Setting  $\varphi(\gl) = \arctan(1/\gl)$, $\gl \ge 0$ and noting
that $\cos(\varphi(\gl)) = {\gl}{(1 + \gl^2})^{-1/2},$ we find
$\cos(\theta(\gl)/2) = {2}^{-1/2}(1 + \gl^2)^{-1/4}(\gl + \sqrt{1
+ \gl^2})^{-1/2}.$ Substituting this expression in \eqref{6.23A}
{ yields}
       \be\la{6.30a}
\re(\sqrt{i - T}) = {2}^{-1/2}(T + \sqrt{1 + T^2}\
)^{-1/2}.
       \ee
Similarly, taking into account $\sin(\theta(\gl)/2) =
\cos(\varphi(\gl)/2)$ and $\cos(\varphi(\gl)/2) = {2}^{-1/2}(1 +
\gl^2)^{-1/4}(\gl + \sqrt{1 + \gl^2})^{1/2}$, we get
        \be\la{6.35}
\im(\sqrt{i- T}) = \int^\infty_0 \sqrt[4]{1 +
  \gl^2}\cos(\varphi(\gl)/2)dE_T(\gl) = \frac{1}{\sqrt{2}}\sqrt{T + \sqrt{1 + T^2}}.
         \ee
It follows from  \eqref{8.7} with account of \eqref{6.30a} and \eqref{6.35} that
 $M(0):=\slim_{x\downarrow 0}M(-x) =: B^K$ { where $B^K$ is} defined
by \eqref{6.32x}. Therefore, by \cite[Proposition 5(iv)]{DM91} the
Krein extension $A^K$ is given by $A_{B^K}:=A^*\upharpoonright
\ker(\gG_1 - B^K\gG_0).$ The remaining statement follows from Proposition \ref{prop6.8}(i) and
Theorem \ref{prop7.2A}(iii).

(iv) The assertion follows from  Theorem \ref{prop7.2A}(iv)
and \eqref{8.23}.
\end{proof}

Next we  generalize Theorem \ref{VII.3} to the case of
unbounded $T\ge0$.
\bt\la{VI.10a}
Let $T = T^* \ge 0$ and $t_1 := \inf\gs_{\rm ess}(T)$.
Further, let $A$ be the minimal operator of $\cA$,
cf. \eqref{8.10}-\eqref{8.11}, and $\wt A = \wt A^* \in Ext_A$. Then

\item[\rm\;\;(i)]
the absolutely continuous part
$\wt A^{ac} E_{\wt A}([t_1,\infty))$ is unitarily
equivalent to the part $A^DE_{A^D}([t_1,\infty))$;

\item[\rm\;\;(ii)]
the Dirichlet, Neumann and Krein  realizations are
$ac$-minimal and $\gs(A^D) = \gs(A^N) = \gs_{ac}(A^K) \subseteq
\gs_{ac}(\wt A)$;

\item[\rm \;\;(iii)]
the $ac$-part $\wt A^{ac}$
is unitarily equivalent to $A^D$ if either $(\wt A - i)^{-1} -
(A^F - i)^{-1} \in \gotS_\infty(\gotH)$ or $(\wt A - i)^{-1} -
(A^K - i)^{-1} \in \gotS_\infty(\gotH)$.
\et
\begin{proof}
By \cite[Corollary 4.2]{MN2011} it  suffices to assume that the
extension $\wt A = \wt A^*$  is disjoint with $A_0,$ that is,
it admits a representation  $\wt A = A_B$ with $B\in \cC(\cH).$

(i) { We} consider the boundary triplet $\wh \gP = \{\cH,\wh
\gG_0,\wh \gG_1\}$ defined { in Lemma \ref{lem6.9}}. In accordance with
\eqref{2.5} the Weyl function corresponding to $A_B$
is given by  $\wh M_B(z) = (B - \wh M(z))^{-1}$, $z \in
\C_+$, where $\wh M(z)$ is given by \eqref{8.7}. Clearly,
\be\la{6.32AB}
\im(\wh M_B(z)) = \wh M_B(z)^*\im(\wh M(z))\wh M_B(z), \qquad z
\in \C_+.
\ee
{ It follows from \eqref{8.7} that  $\bigl(\re(\sqrt{i -
T})\bigr)^{-1}\ge \sqrt 2.$ Therefore \eqref{6.30a} yields}
\be\la{6.32a}
 \im(\wh M(z)) \ge \sqrt{2}\im(M(z)), \quad z \in
\C_+, \quad \text{where} \quad M(z) = i\sqrt{z - T},
\ee
cf.  \eqref{8.2B}. Following the line of reasoning of the proof of
Theorem \ref{VII.3}(i) we obtain from \eqref{6.32a} that $d_{\wh
M^D}(t)=\infty$ for a.e. $t\in [t_1, \infty),$ where  $D=D^*\in
\gotS_2(\cH)$ and $\ker D=\{0\}.$ Moreover, it follows from
\eqref{6.32AB} that $d_{\wh M_B^D}(t)= d_{\wh M^D}(t)=\infty$ for
a.e. $t\in [t_1, \infty).$ One completes the proof by applying
Theorem \ref{IV.10}.

(ii) To prove (ii) for $A^D$ we use  again  estimates \eqref{6.32a} and
follow the proof of Theorem \ref{VII.3}(ii). We complete the proof
for $A^D$ by applying Theorem \ref{IV.10}. Taking into account Theorem \ref{VI.9}(iv)
we complete the proof of (ii).

(iii) The Weyl function $\wh M(\cdot)$
is given by  \eqref{8.7}. Taking into account
\eqref{8.7C} one obtains  $\sup_{n\in\N}\gotm^+_n < \infty$,
where $\gotm^+_n$ is the invariant maximal normal function defined by
\eqref{4.12A}. Indeed, this follows from
\eqref{7.3a} because this estimate shows that $\gotm^+_n$ does not
depend on $n\in \N$. Applying Theorem \ref{V.5}
we complete the proof.

To prove the second statement we note that the operator
$B^K$ defined by \eqref{6.32x} is  bounded. Therefore, by
\eqref{2.5} to $A_{\!\!B^K}$ the Weyl function
\bed
\wh M_{\!\!B^K}(z) = (B^K - \wh M(z))^{-1}, \quad z \in \C_+.
\eed
corresponds. Inserting expression \eqref{6.32x} into this formula we get
\bed { \wh M_{\!\!B^K}(z)} = -\frac{1}{\sqrt{2}}\frac{1}{\sqrt{T} +
i\sqrt{z-T}}\frac{1}{\sqrt{T + \sqrt{1 + T^2}}}=
\frac{1}{z\sqrt{2}}\frac{\sqrt{T} - i\sqrt{z - T}}{\sqrt{T +
\sqrt{1 + T^2}}}.
 \eed
It follows that the limit ${ \wh M_{B^K}(t+i0)}$ exists for  any
$t\in \R \setminus \{0\}$ and
\bed
{ \wh M_{\!\!B^K}(t)} := \slim_{y\to*0}M_{\!\!B^K}(t+iy) =
-\frac{1}{t\sqrt{2}}\frac{\sqrt{T} - i\sqrt{t - T}}{\sqrt{T + \sqrt{1 + T^2}}}.
\eed
Clearly, $\wh M_{\!\!B^K}(t)\in [\cH]$ for any $t\in \R \setminus \{0\}.$
By  Theorem \ref{V.5} the  $ac$-parts of $\wt A$ and $A^K$
are unitarily equivalent whenever $(\wt A - i)^{-1} - (A^K -
i)^{-1} \in \gotS_\infty(\gotH).$ This completes the {
proof.}
\end{proof}

Finally, we generalize Corollary \ref{VI.5Z} to unbounded
operator potentials.
\bc\label{cor6.12}
Let the assumptions of Theorem \ref{VI.10a} be satisfied. If, in
addition,  $\dim(\cH) = \infty$ and $t_0:=\inf\gs(T) =
\inf\gs_{\ess}(T)=:t_1$, then

\item[\;\;\rm (i)]
the Dirichlet, Neumann and Krein realizations are
strictly $ac$-minimal;

\item[\;\;\rm (ii)]
the $ac$-part $\wt A^{ac}$ of $\wt A$ is unitarily equivalent to $A^D$
whenever \eqref{4.18A} is satisfied.
\ec
\begin{proof}
The first statement is immediately follows 
from Theorem \ref{VI.10a}(i) and Theorem \ref{VI.9}(iv). The
second statement is proved in just the same way as Corollary
\ref{VI.5Z}(ii).
\end{proof}

Next we apply Theorem \ref{VI.10a} to realizations $A_C$ of
the form
\bed
y'(0) = Cy(0),\qquad C = C^*\in \cC(\cH),
\eed
using results of the papers \cite{GorKut82,GorKut78}.
\begin{corollary}
Let the assumptions of Theorem \ref{VI.10a} be satisfied. If either
\begin{equation}\label{5.18a}
(T+I)^{-1/2}C(T+I)^{-1/2}\in\mathfrak S_\infty(\cH)\quad
\text{or}\quad (T+I)^{-1}\in\mathfrak S_\infty(\cH)
\end{equation}
is valid, then the $ac$-part
$A_C^{ac}$ is unitarily equivalent to $A^D$.
\end{corollary}
\begin{proof}
According to  \cite{GorKut78, GorKut82}  $(A_C -i)^{-1} -
(A^N - i)^{-1} \in\mathfrak S_p$ provided that
either $(T+I)^{-1/2}C(T+I)^{-1/2}\in\mathfrak S_p(\cH)$ or
$(T+I)^{-1}\in\mathfrak S_p(\cH)$ for $p\in [1,\infty]$. It remains to
apply Theorem \ref{VI.10a}(iii).
\end{proof}
\begin{remark}
{\em
Clearly,   $t_0 \not = t_1$ if  $T\ge 0$ and $(T+I)^{-1}\in \mathfrak
S_\infty.$ Thus, in this case $A^D$  is $ac$-minimal but not
strictly $ac$-minimal.
}
\end{remark}

\subsection{Application}

In this subsection we apply previous results to
Schr\"odinger operators in the half-space. To this end we denote by
$L = L_{\rm min}$ the minimal elliptic operator associated with the differential expression
\bed
\cL := -\frac{\partial^2}{\partial t^2} - \sum^n_{j=1}\frac{\partial^2}{\partial x^2} + q(x),
\quad q(x) = \overline{q(x)} \in L^\infty(\R^n),
\eed
in $L^2(\R^{n+1}_+)$, $\R^{n+1}_+:=\R_+\times \R^n$.
Recall that $L_{\min}$ is the closure of $\cL$
defined on $C^\infty_0(\R^{n+1}_+)$.
It holds $\dom(L_{\min}) = H^2_0(\R^{n+1}_+) := \{f \in H^{2}(\R^{n+1}_+):
f\upharpoonright\partial\R^{n+1}_+ = 0, \quad
\frac{\partial f}{\partial \gotn}\upharpoonright\partial\R^{n+1}_+ = 0\}$
where $\gotn$ stands for the interior normal to $\partial\R^{n+1}_+$. Clearly, $L$
is symmetric. The maximal operator $L_{\max}$ is defined by
$L_{\max} = (L_{\min})^*$. We emphasize that $H^{2}(\R^{n+1}_+)\subset
\dom(L_{\max})\subset H^{2}_{loc}(\R^{n+1}_+)$ but
$\dom(L_{\max})\not = H^{2}(\R^{n+1}_+)$. The trace mappings
$\gamma_j\colon C^{\infty}({\overline{\R^{n+1}_+}})
\longrightarrow  C^{\infty}(\partial\R^{n+1}_+)$,
$j\in \{0,1\}$ are defined by
$\gamma_0 f := f \upharpoonright\partial\R^{n+1}_+$  and
$\gamma_1 f := \frac{\partial f}{\partial \gotn}\upharpoonright\partial\R^{n+1}_+$.
Let $\gotL_+$ be the domain $\dom(L_{\max})$
equipped with the graph norm. It is known (see \cite{Gru08,LioMag72})
that $\gamma_j$ can be extended by continuity to the
operators mapping $\gotL_+$ continuously onto
$H^{-j-1/2}(\partial\R^{n+1}_+), \ j\in \{0,1\}.$

Let us define the following  realizations of $\cL$:
\begin{enumerate}

\item[(i)]
$L^Df := \cL f$,  $f \in \dom (L^D) :=
\{\varphi  \in H^2(\R^{n+1}_+):\  \gamma_0 \varphi = 0\}$;

\item[(ii)]
$L^Nf := \cL f$,  $f \in \dom (L^N) :=
\{\varphi \in H^2(\R^{n+1}_+): \gamma_1 \varphi=0\}$;

\item[(iii)]
$L^Kf := \cL f$, $f \in \dom (L^K) :=
\{\varphi \in \dom(L_{\max}): \gamma_1 \varphi + \Lambda \gamma_0
\varphi = 0\}$ where $\Lambda := \sqrt{
-\Delta_x + q(\cdot)}:\   H^{-1/2}(\partial\R^{n+1}_+)\to
H^{-3/2}(\partial\R^{n+1}_+).$
\end{enumerate}

To treat the operator $L_{\min}$ as the Sturm-Liouville operator
with (unbounded) operator potential we denote by $T$  the
minimal operator associated  with the Schr\"odinger expression
\begin{equation}\label{7.41}
\cT := -\Delta_x  + q(x) := - \sum_{j=1}^n\frac{\partial^2}{\partial
x_j^2} + q(x), \quad \overline{q(x)} = q(x),
\end{equation}
in $\cH := L^2({\mathbb R}^n)$. It turns out that $T$ is
Moreover, If $q(x) \ge 0$, then $T\ge 0$.
Let $A := A_{\rm min}$ be the minimal operator
associated with \eqref{8.10} where $T = T_{\rm min}$.
\bp\label{prop6.12}
Let $q(\cdot) \in L^\infty(\R),$ \ $q(\cdot) \ge 0$,  and let $T$
be the minimal (self-adjoint) operator associated with $\cT$ in $L^2(\R)$.
Let also $t_0:=\inf \gs(T)$ and $t_1:=\inf \sigma_{\ess}(T)$.
Then:

\item[\rm \;\;(i)]
the  minimal operator $A$ coincides with
the minimal operator $L$ and   $\dom(A)= H^2_0(\R^{n+1}_+)$;

\item[\rm\;\;(ii)]
the Dirichlet realization  $A^D$ coincides with
$L^D$, hence, $L^D$ is absolutely continuous and
$\sigma(L^D)=\sigma_{ac}(L^D) = [t_0,\infty)$;

\item[\rm\;\;(iii)]
the  Neumann realization $A^N$
coincides with $L^N$, in particular, $L^N$ is absolutely
continuous and $ \sigma(A^N)  = \sigma_{ac}(A^N)=[t_0, \infty)$;

\item[\rm\;\;(iv)]
the Krein realization $A^K$ coincides with
$L^K$, in particular, $L^K$ admits the decomposition $L^K =
0_{\cH_0}\bigoplus (L^K)^{ac}$, $\cH_0 := \ker(L^K)$,  and
$\sigma_{ac}(L^K)=[t_0,\infty)$;

\item[\rm\;\;(v)]
the self-adjoint realizations  $L^D$,\  $L^N,$\  and $L^K$ are
  $ac$-minimal, in particular, $L^D$,\  $L^N,$\  and
$(L^K)^{ac}$ are unitarily equivalent to each other. If
$t_0=t_1$,  then  the operators $L^D$, $L^N$ and $L^K$ are
strictly $ac$-minimal;

\item[\rm\;\;(vi)]
if $\wt L$ is a  self-adjoint realization of $\cL$ such that either
$({\wt L}-i)^{-1}-(L^D-i)^{-1}\in \gotS_{\infty}(L^2(\R^{n+1}_+))$ or
$({\wt L}-i)^{-1}-(L^K-i)^{-1}\in \gotS_{\infty}(L^2(\R^{n+1}_+))$ is
satisfied, then ${\wt L}^{ac}$ and $L^D$ are unitarily
equivalent;

\item[\rm\;\;(vii)]
If $t_0=t_1$ and if $\wt L$ is a  self-adjoint realization of $\cL$ such
that $(\wt L - i)^{-1} - (L^N - i)^{-1} \in
\gotS_\infty(L^2(\R^{n+1}_+))$ is satisfied, then ${\wt L}^{ac}$ and $L^D$ are unitarily
equivalent.
\end{proposition}
\begin{proof}
(i) We introduce the set
\bed
\cD_\infty := \left\{\sum_{1\le j \le k} \phi_j(x)h_j(\xi):
\;\phi_j \in C^\infty_0(\R_+), \  h_j \in C^\infty_0(\R^n),\  k \in
\N\right\}
\eed
We note that $\cD_\infty \subseteq \cD_0$, which is  given by \eqref{8.11}, and $\cD_\infty
\subseteq C^\infty_0(\R^{n+1}_+)$. Moreover, $A
\upharpoonright\cD_\infty = L\upharpoonright\cD_\infty$.  Since
$\cD_\infty$ is a core for both minimal operators $A$ and $L$
we have $A = L$ which yields $\dom(A) = H^{2,2}_0(\R^{n+1}_+)$.

(ii) Since $A = L$ we have $A^F = L^F$.
Using $L^F = L^D$ the proof of (ii) follows immediately from
Theorem \ref{VI.9}(i).

(iii) One verifies that $W^{2,2}_T(\R_+,\cH) = H^2(\R^{n+1}_+)$, i.e,
both spaces are isomorphic. A straightforward computation shows
that
\bed
\gt^\cL[f] := (\cL f,f)_\gotH = (\cA f,f)_{L^2(\R^{n+1}_+)} =: \gt^\cA[f],
\quad f \in W^{2,2}_T(\R_+,\cH) = H^2(\R^{n+1}_+).
\eed
Since $W^{2,2}_T(\R_+,\cH)$ is dense in $W^{1,2}_{\sqrt{T}}(\R_+,\cH)$ the
completion of $\gt^\cA$ gives $\gt_N$ defined by \eqref{4.18} which is the closed
quadratic form associated with $A^N$.
Moreover, using that $H^{2,2}(\R^{n+1}_+)$ is dense in $H^{1,2}(\R^{n+1}_+)$
the completion of $\gt^\cL$ gives the closed quadratic form
associated with $L^N$. Since both completion coincide we get that
$A^N = L^N$. The remaining part follows from Theorem \ref{VI.9}(ii).

(iv) Since $A = L$ we have that $A^K$ is identical
with the Krein realization of $\cL$. However, it was proved
in \cite[Section 9.7]{DM91} that even
$L^K$ is the Krein extension of $\cL$  The rest of the
statements is implied by Theorem \ref{VI.9}(iii).

(v) By Theorem \ref{VI.10a}(ii) the extension
$A^D$, $A^N$ and $A^K$  are $ac$-minimal. Taking into account
(i) - (iv) we find that $L^D$, $L^N$ and $L^K$ are $ac$-minimal.
The second statement of (v) follows from Corollary \ref{cor6.12}(i).

(vi) This statement follows immediately from Theorem \ref{VI.10a}(iii)
and (ii).

(vii) It follows from Corollary \ref{cor6.12}(ii).
\end{proof}
\begin{remark}
{\em
Let  $T$ be the (closed) minimal non-negative operator associated
in $\cH := L^2({\mathbb R}^n)$ with general uniformly elliptic
operator
\bed
\wt{\cT} := - \sum_{j,k=1}^n \frac{\partial}{\partial x_j} a_{jk}(x)
\frac{\partial}{\partial x_j} + q(x), \;\; a_{jk} \in
C^{1}({\overline\R^{n+1}_+}),\ \ q \in C({\overline\R^{n+1}_+})\cap
L^\infty({\R^{n+1}_+}),
\eed
where the coefficients $a_{jk}(\cdot)$ are bounded with their
$C^1$-derivatives, $q\ge 0.$ If the coefficients have some
additional "good" properties, then $\dom(T)=H^2({\R}^n)$
algebraically and topologically.  By Lemma \ref{VI.2.7},
{ $\dom(A_{\rm \min})= W^{2,2}_{0,T}(\R_+,\cH) =
H^{2,2}_0(\R^{n+1}_+)$ } and Proposition  \ref{prop6.12} remains valid with
$T$ in place of the Schr\"odinger operator \eqref{7.41}.

Note also that the Dirichlet and the Neumann realizations $L^D$
and $L^N$ are always self-adjoint ((cf. \cite[Theorem
2.8.1]{LioMag72}, \cite{Gru08})).
}
\end{remark}
\begin{corollary}\label{cor6.13}
Let the assumptions of Proposition  \ref{prop6.12} be satisfied. If
\begin{equation}\label{6.33}
\lim_{|x|\to\infty}\int_{|x-y|\le 1}q(y)dy = 0,
\end{equation}
then the  realizations $L^D$, $L^N$ and $L^K$ are strictly
$ac$-minimal and
\bed
\gs(L^D) = \gs_{ac}(L^K) = \gs(L^N) =\gs_{ac}(L^N)
= [0,\infty).
\eed
\end{corollary}
\begin{proof}
By \cite[Section 60]{Gla66}  condition  \eqref{6.33} yields the
equality $\sigma_c(T)= \R_+,$ in particular $0\in \sigma_c(T)$ and
$t_1=0.$ Since $q\ge 0,$\  we have $0\le t_0 \le t_1=0,$ that is
$t_0 = t_1=0$. It remains to apply Proposition \ref{prop6.12}(i)-(iv).
\end{proof}
\begin{remark}
{\em
Condition  \eqref{6.33} is satisfied whenever
$\lim_{|x|\to\infty}q(x) = 0.$  Thus,  in this case the
conclusions  of Corollary  \ref{cor6.13}  are valid.
However, it might happen that  $\gs(L^D) = \gs(L^N) = \gs_{ac}(L^K)
= [t_0,\infty)$, $t_0 > 0$, though  $\inf q(x) = 0$.
}
\end{remark}

\setcounter{section}{0}
\renewcommand*\thesection{\Alph{section}}
\section{Appendix: Operators admitting separation of variables}

\subsection{Finite interval}

Here we consider the differential expression $\cA$ with unbounded
$T = T^* \ge 0$ (cf.  \eqref{8.10})  on a finite interval $I =
[0,\pi]$ and denote it by $\cA_I$. The minimal operator $A
:=A_{I,\rm \min} := \overline{A'}$ generated by $\cA$ in the
Hilbert space $\gotH_I := L^2(I,\cH)$ is  defined similarly to
that of $A=A_{\min}$ in $L^2(\R_+,\cH)$. Obviously,  $A_{I,\rm
min}$ is densely defined and non-negative.

We briefly discuss the spectral properties of realizations of
$\cA_I$ which  admit separating of variables. We set
\bead
A^D_If & :=&  \cA_If, \quad f \in \dom(A^D_I) := \{f \in W^{2,2}_T(I,\cH): f(0) = f(\pi) = 0\}\\
A^N_If & :=&  \cA_If, \quad f \in \dom(A^D_I) := \{f \in W^{2,2}_T(I,\cH): f'(0) = f'(\pi) = 0\}
\eead
where $W^{2,2}_T(I,\cH) = W^{2,2}(I,\cH) \cap L^2(I,\cH_1(T))$ with
$\cH_1(T)$ defined by \eqref{5.1a}.

 To state the main result  denote by $l_D$ and $l_N$ the Dirichlet and
Neumann realization of the differential expression $l := -d^2/dx^2$ in
the Hilbert space $L^2(I)$, i.e.
\bed
\begin{matrix}
l_D & := & -\frac{d^2}{dx^2}\upharpoonright\dom(l_D), \;\dom(l_D)=\{f\in
W^{2,2}[0,\pi]:f(0)=f(\pi)=0\},\\
l_N & := & -\frac{d^2}{dx^2}\upharpoonright\dom(l_N), \;\dom(l_N)=\{f\in
W^{2,2}[0,\pi]:f'(0)= f'(\pi)=0\}.
\end{matrix}
\eed
Obviously, both spectra are discrete and given by $\gs(l_D) =
\{1,4,\ldots,k^2,\ldots\}$, $k \in \N$ and $\gs(l_N) =
\{0,1,4,\ldots,k^2,\ldots\}$, $k \in \N_0 := \{0\} \cup \N$.
\bp\label{propAp1}
Let $A^D_I$ and $A^N_I$ be the Dirichlet and the Neumann
realizations of $\cA_I$ in $L^2(I,\cH)$ and let
$T_k:=T + k^2 I_{\cH}\bigl(\in\cC(\cH)\bigr)$. Then

\item[\;\;\rm (i)]
$A^D_I$ is unitarily equivalent to the operator
$\oplus^{\infty}_{k=1}T_k$;

\item[\;\;\rm (ii)]
$A^N_I$ is unitarily equivalent to the operator $\oplus^{\infty}_{k=0}T_k$;

\item[\;\;\rm (iii)]
The spectrum of the operators $A^D_I$ and  $A^N_I$ is discrete, pure
point, purely singular and absolutely continuous if and only if the
spectrum of $T$ is so.

\item[\;\;\rm (iv)]
The spectral multiplicity functions $N_{A^D_I}(\cdot)$ and $N_{A^D_N}(\cdot)$
of the realizations $A^D_I$ and $A^N_I$, respectively, are finite for
each $\gl \in \R$ whenever the multiplicity function $N_{T}(\cdot)$ is finite.
Moreover, if $\sigma_{ac}(T)=[t_0,\infty)$,
then $\sigma_{ac}(A^D_I) = [t_0+1, \infty)$ and
\bed
N_{(A^D_I)^{ac}}(t)=pN_{T^{ac}}(t)\quad \text{for a.e.}\quad
t\in[t_0+k^2,t_0+(k+1)^2),\quad k\in \N,
\eed
as well as $\sigma_{ac}(A^D_I) = [t_0, \infty)$ and
\bed
N_{(A^N_I)^{ac}}(t) = (p +1) N_{T^{ac}}(t)\quad \text{for a.e.}\quad
t\in[t_0+k^2, t_0 + (k+1)^2),
\eed
$k \in \N_0 := \{0\} \cup \N$.

\item[\;\;\rm (v)]
The operators $(A^D_I)^{ac}$ and $(A^N_I)^{ac}$ are not
unitarily equivalent.
\end{proposition}
\begin{proof}
(i) By the spectral theorem, the operator $l_D = l_D^*$ is
unitarily equivalent to the diagonal operator
$\Lambda_D=\diag(1^2,2^2,\ldots, k^2,\ldots)$ acting in $\mathfrak
H_D =l^2(\N)$. Namely, $U_D l_D = \Lambda_D U_D$ where $U_D$ is
the unitary map from $L^2[0,\pi]$ onto $l^2(\mathbb N)$,
\bed
U_D:\  f=\sqrt{\frac{2}{\pi}}\sum^{\infty}_{k=1}a_k\sin
kx\to\{a_k\}^{\infty}_1\in l^2(\mathbb N)
\eed
and $a_k=(f,\sqrt{2/\pi}\sin kx)$. Hence
\bead
\lefteqn{\hspace{-10mm}
(U_D\otimes I_{\cH}) A^D(U^*_D\otimes I_{\cH})  =
(U_D\otimes
I_{\cH})(l_D\otimes I_{\cH} +  I_{\mathfrak H_1}\otimes T)(U_D^*\otimes
I_{\cH}) = }\\
& &
\Lambda_D\otimes I_\cH +I_{\mathfrak H_2}\otimes T
= \bigoplus^{\infty}_{k=1}(k^2I_{\cH} + T) =
\bigoplus^{\infty}_{k=1}T_k.
\eead

(ii) In this case, by the spectral theorem, the operator $A^N$ is
unitarily equivalent to the diagonal operator
$\Lambda_N=\diag(0,1^2,2^2, \ldots,k^2,\ldots)$ in
$\gH_N =l^2(\N_0)$, $U_N l_N=\Lambda_N U_N$ where
\bed
U_N:
 f=\frac{1}{\sqrt{\pi}}b_0+\sqrt{\frac{2}{\pi}}\sum^{\infty}_{k=1}b_k\cos
kx\to\{b_k\}^{\infty}_0\in l^2(\N_0)
\eed
and $b_k=(f,\sqrt{2/\pi}\cos kx)$. Repeating the previous
reasonings we arrive at the required relation
\bed
(U_N\otimes I_{\cH})A^N(U^*_N\otimes I_{\cH}) =
\oplus^{\infty}_{k=0}T_k.
\eed

(iii) This statement follows immediately from (i) and (ii) in view of the
obvious relations $\gs\bigl(\bigoplus^{\infty}_{k=1}T_k\bigr) =
\bigcup_{k=1}^\infty \gs(T_k)$ and
$\gs_{\tau}\bigl(\bigoplus^{\infty}_{k=1}T_k\bigr) = \bigcup_{k=1}^\infty
\gs_{\tau}(T_k)$,  $\tau= pp,s,sc,ac$.

(iv) From (i) and (ii) and the
obvious relations $\sigma_{\tau}(T_k) = k^2 + \sigma_{\tau}(T_k)$,
$\tau= d, pp,s,sc,ac$, $k\in \N$ we verify (iv).

(v) From (i) and (ii) it follows that
$\gs_{ac}(A^N_I) = \bigcup^\infty_{k=0}\gs_{ac}(T_k)$
and $\gs_{ac}(A^D_I) = \bigcup^\infty_{k=1}\gs_{ac}(T_k)$
which yields $\gs_{ac}(A^N_I) \not= \gs_{ac}(A^D_I)$
which proves (v).
\end{proof}

\subsection{Semi-axis}\la{A.2}

Our next purpose is to show that the spectral properties of
realizations of $\cA$ admitting separation of variables can be
investigated directly by applying elementary methods. In
particular, we present a simple proof of Theorem
\ref{VI.9}(ii). let us at first prove a general statement.
\bl\la{A.2.I}
Let $K$ and $T$ be self-adjoint operators in the separable Hilbert spaces $\cK$
and $\cH$, respectively, and let $L_K := K \otimes I_\cH + I_\cK
\otimes T$ which is self-adjoint in $\cK \otimes \cH$.

\item[\;\;\rm (i)]
If the self-adjoint operators $K_1$ and $K_2$ are unitarily
equivalent, then $L_{K_1}$ and $L_{K_2}$ are unitarily equivalent

\item[\;\;\rm (ii)]
If $K$ is absolutely continuous, then $L_K$ is absolutely continuous.
\el
\begin{proof}
(i) Let $V$ be a unitary operator such that $K_2 = V^*K_1V$. Then
$U := V \otimes I_\cH$ is unitary and
\bed
U^* L_{K_1} U = V^* \otimes I_\cH ( K_1 \otimes I_\cH + I_\cK
\otimes T)V \otimes I_\cH = K_2 \otimes I_\cH + I_\cK \otimes T
= L_{K_2}.
\eed

(ii) Let $\goth$ be an auxiliary infinite dimensional separable
Hilbert space. In $L^2(\R,\goth)$ we consider the multiplication operator $Q$ defined by
\be\la{A.2.1}
(Qf)(t) = tf(t), \quad t \in \R, \quad f \in L^2(\R,\goth).
\ee
If $K$ is absolutely continuous, then there is an isometry
$\Phi_0 :\cK \longrightarrow L^2(\R,\goth)$ such that
$Q \Phi_0 = \Phi_0 K$, $\Phi^*_0\Phi_0 = I_\cK$.
Hence the isometry $\Phi := \Phi_0 \otimes I_\cH:
\cK \otimes \cH  \longrightarrow L^2(\R,\goth) \otimes \cH$
intertwines $L_K$ and $\wh L := Q \otimes I_\cH +
I_{L^2(\R,\goth)} \otimes T$, i.e.
\bed
\wh L \Phi = \Phi L_K.
\eed
Notice that $L^2(\R,\goth) \otimes
\cH = L^2(\R,\goth \otimes \cH)$. The operator $\wh L$ has in $L^2(\R,\goth')$,
$\goth' :=\goth \otimes \cH$, the representation
$\wh L := \wh Q + \wh T$ where $\wh Q$ is a multiplication operator
which  is defined similarly as $Q$, cf. \eqref{A.2.1}, and $\wh T$ is given by
\bed
(\wh T f)(t) := T'f(t), \quad f \in \dom(\wh T) := \{f
  \in L^2(\R,\goth'): T' f(t) \in L^2(\R,\goth')\}
\eed
where $T' := I_\goth \otimes T$. Using the Fourier
transform $\cF$ one easily verifies that $\wh Q$ is
unitarily equivalent to the momentum operator $-i\frac{d}{dt}$ in
$L^2(\R,\goth')$, i.e $\cF^{-1} \wh Q \cF = -i\frac{d}{dt}$. This
yields that
\bed
\cF \wh L  \cF^{-1} = -i\frac{d}{dt} + \wh H.
\eed
Finally, using the gauge transform $(\cG f)(t) = e^{-it\wh H}f(t)$, $f
\in  L^2(\R,\goth')$, we find  $\cG\cF \wh L  \cF^{-1} \cG^{-1} = -i\frac{d}{dt}$. Hence
\be\la{A.2.2}
-i\frac{d}{dt} \;\cG\cF\Phi = \cG\cF\Phi L_K
\ee
Since the momentum operator $-i\frac{d}{dt}$ is absolutely continuous
the relation \eqref{A.2.2} immediately implies that $L_K$ is absolutely
continuous.
\end{proof}

We consider the self-adjoint operator
\bed
l_\tau := -\frac{d^2}{dt^2}\upharpoonright\dom(l_\tau), \qquad \dom(l_\tau)=\{f\in
W^{2,2}({\mathbb R}_+):\ f'(0)= \tau f(0)\},
\eed
in $\cK := L^2(\R_+)$ where $\tau \in \R_+\cup\{0\}\cup\{\infty\}$.
The extensions $\tau = 0$ and $\tau = \infty$ are identified with
the Neumann and the Dirichlet realizations of $-\frac{d^2}{dt^2}^2$, respectively. Further,
let  $T=T^*\ge 0$, $T\in\cC(\cH)$. Consider the family of self-adjoint
operators
\be\la{A.3}
A_\tau := l_\tau \otimes I_{\cH}+I_{\cK}\otimes T,
\qquad \tau \in \R_+\cup\{0\}\cup\{\infty\},
\ee
in the Hilbert space $\cK \otimes \cH = L^2(\R_+,\cH)$. Note
for each $\tau \in \R_+\cup\{0\}\cup\{\infty\}$ the
operator $A_\tau$ can be regarded as a self-adjoint extension of the
minimal operator $A$ defined by \eqref{8.10} and \eqref{8.11}. In
particular, we have $A_0 = A^N$ and $A_\infty = A^D$.
\bc\la{A.2.II}
Let $T = T^* \ge 0$.

\item[\;\;\rm (i)] If $\tau_1 \ge 0$ and $\tau_2 \ge 0$, then
  $A_{\tau_1}$ and $A_{\tau_2}$ are unitarily equivalent. In particular,
  the extensions $A^D$ and $A^N$ are unitarily equivalent.

\item[\;\;\rm (ii)] If $\tau \ge 0$, then $A_\tau$ is absolutely
  continuous. In particular, $A^D$ and $A^N$ are absolutely
  continuous.
\ec
\begin{proof}
(i) From \cite[Section 21.5]{Nai69} we get that the operators $l_\tau$
are unitarily equivalent to each other if $\tau \ge 0$. Applying Lemma
\ref{A.2.I}(i) we prove (i).

(ii) Using the Fourier transformation one easily proves that the
operator $l_0$ is absolutely continuous. Taking into account Lemma
\ref{A.2.I}(ii) we verify (ii).
\end{proof}
\br\la{A.2.III}
{\em
\item[\;\;\rm (i)]
We note that the above reasonings cannot be applied to
realizations of $\cA$ which do not admit the tensor product structure
\eqref{A.3}.

\item[\;\;\rm (ii)]
 Comparing Corollary  \ref{A.2.II} with Proposition
\ref{propAp1} we obtain that there  are substantial differences between spectral
properties of realizations on the semi-axis $\R_+$ and on a finite
interval $I$. Indeed, for self-adjoint realizations of $\cA$ on $\R_+$
the $ac$-part can never be eliminated for any
$T=T^*\ge 0$, cf. Theorem \ref{VI.10a}(ii). In contrast to that
the spectral properties of self-adjoint
realizations of $\cA_I$ strongly depend on $T$.
}
\end{remark}


\def\cprime{$'$} \def\cprime{$'$} \def\cprime{$'$} \def\cprime{$'$}
  \def\cprime{$'$} \def\lfhook#1{\setbox0=\hbox{#1}{\ooalign{\hidewidth
  \lower1.5ex\hbox{'}\hidewidth\crcr\unhbox0}}} \def\cprime{$'$}
  \def\cprime{$'$} \def\cprime{$'$} \def\cprime{$'$} \def\cprime{$'$}
  \def\cprime{$'$} \def\cprime{$'$} \def\cprime{$'$}
  \def\lfhook#1{\setbox0=\hbox{#1}{\ooalign{\hidewidth
  \lower1.5ex\hbox{'}\hidewidth\crcr\unhbox0}}}

\end{document}